\documentclass[10pt]{article}
\usepackage{graphicx}
\usepackage{xcolor} 
\usepackage{amsthm,amssymb,amsmath}
\usepackage{mathrsfs}
\usepackage{relsize}
\usepackage[pdfa]{hyperref}

\usepackage[utf8]{inputenc}

\newtheorem{thm}{Theorem}[section]
\newtheorem{lem}[thm]{Lemma}
\newtheorem{pro}[thm]{Proposition}
\newtheorem{cor}[thm]{Corollary}

\newtheorem{rema}{Remark}[section]

\numberwithin{equation}{section}

\def\ep{\varepsilon}
\DeclareMathOperator{\E}{\mathbb{E}}

\def\L{\mathcal{L}}
\def\F{\mathcal{F}}
\def\hep{\widehat{\ep}}
\def\n{n}
\def\wt{\widetilde}
\def\wh{\widehat}

\def\ol{\overline}
\def\F{\mathcal{F}}
\def\I{\mathbf{I}}
\def\dz{\mathbf{d}^2z}
\DeclareMathOperator{\Tr}{Tr}
\def\tr{\mathrm{tr}_n}
\DeclareMathOperator{\Int}{Int}

\DeclareMathOperator{\supp}{supp}
\DeclareMathOperator{\Prob}{Prob}
\DeclareMathOperator{\dist}{dist}

\makeatletter
\newsavebox{\@brx}
\newcommand{\llangle}[1][]{\savebox{\@brx}{\(\m@th{#1\langle}\)}%
  \mathopen{\copy\@brx\kern-0.5\wd\@brx\usebox{\@brx}}}
\newcommand{\rrangle}[1][]{\savebox{\@brx}{\(\m@th{#1\rangle}\)}%
  \mathclose{\copy\@brx\kern-0.5\wd\@brx\usebox{\@brx}}}
\makeatother

\begin{document}

\title{
\fbox{\parbox{0.95\textwidth}{\normalsize This preprint has not undergone
peer review or any post-submission improvements or corrections. The Version of Record of this article is
published in Journal of Statistical Physics, and is available online at \url{https://doi.org/10.1007/s10955-025-03492-z}}}\bigskip

\quad\bigskip

Limiting eigenvalue distribution of\\ the general deformed Ginibre ensemble }

\author{Roman Sarapin\\
\small{B.Verkin Institute for Low Temperature Physics and Engineering}\\
\small{of the National Academy of Sciences of Ukraine, Kharkiv, Ukraine}
}

\date{}

\maketitle

\setcounter{tocdepth}{1}

\begin{abstract}
Consider the $n\times n$ matrix $X_n=A_n+H_n$, where $A_n$ is a $n\times n$ matrix (either deterministic or random) and $H_n$ is a $n\times n$ matrix independent from $A_n$ drawn from complex Ginibre ensemble.
We study the limiting eigenvalue distribution of $X_n$. In \cite{TVKr:10} it was shown that the eigenvalue distribution of $X_n$ converges to some deterministic measure. This measure is known for the case $A_n=0$. Under some general convergence conditions on $A_n$ we prove a formula for the density of the limiting measure.
We also obtain an estimation on the rate of convergence of the distribution. The approach used here is based on supersymmetric integration.
\end{abstract}

\section{Introduction}

\subsection{General information}

Let $H_n$ be a $n\times n$ random matrix with i.i.d.\ complex Gaussian random entries  $h_{ij}$ satisfying $\E\{h_{ij}\}=0$, $\E\{h^2_{ij}\}=0$ and $\E\{|h_{ij}|^2\}=1/n$. Consider a random $n\times n$ matrix
\begin{equation*}
X_n=A_n+H_n,
\end{equation*}
where $A_n$ is either deterministic or random $n\times n$ matrix with entries independent of $h_{ij}$. 
The matrices $X_n$ form the so called deformed Ginibre ensemble.

Matrices of such form play a significant role in communication theory, where $A_n$ and $H_n$ are considered to be the signal matrix and the noise matrix respectively. For the purposes of that theory, the behaviour of the smallest singular value $\sigma_1(X_n)=(\lambda_1(X_nX_n^*))^{1/2}$ of $X_n$ has been extensively studied, see \cite{SST:06}, \cite{Ed:88}, \cite{RV:08}, \cite{TV:10}, \cite{CiErS:rc}, \cite{CiErS:rc_den}, \cite{TV:07}, \cite{TV:10_1}, \cite{C:18}, \cite{T:20}, \cite{SS:sng} for the details. 

This paper focuses on the limiting eigenvalue distribution of $X_n$. Deriving the limiting distribution of eigenvalues is a fundamental problem in Random Matrix Theory. The pioneering work in this area was done by Wigner \cite{Wig:58} for $n\times n$ hermitian matrices with i.i.d.\ Gaussian random entries. 
This result was later extended to the case of arbitrary i.i.d. entries with mean $0$ and variance $1/n$ (see \cite{Pstr:73}). 
The first hypothesis regarding the eigenvalue distribution of non-hermitian matrices, known as \textit{Circular Law Conjecture}, was posed in 1950's. The conjecture stated that  
for $n\times n$ matrices with i.i.d.\ random entries with mean $0$ and variance $1/n$, the limiting eigenvalue distribution on the complex plane is given by the following density:
\begin{equation*}
\rho(z)=
\begin{cases}
\frac{1}{\pi},\quad& |z|\le 1,\\
0,\quad&|z|>1,
\end{cases}
\end{equation*}
i.e.\ the eigenvalues are distributed uniformly on the unit disk. The case of complex Gaussian entries was established by Mehta \cite{Meh:67} in 1967. However, extending this result to an arbitrary entry distribution proved to be more challenging. In 1984 Girko presented an approach to this problem based on the logarithmic potential theory (see \cite{Gir:84}), but this approach did not let to prove the conjecture in full generality at that time.
Later Circular Law Conjecture was verified in \cite{Ed:88}, \cite{Gir:04}, \cite{Ba:97, BaSil:06}, \cite{GoTi:1, GoTi:2}, \cite{PaZh}, \cite{TV:08} under certain additional assumptions on the entries distribution, and was proven in 2010 by Tao and Vu in \cite{TVKr:10} under the most general assumptions.

The deformed Ginibre ensemble is in fact a generalization of the non-hermitian ensemble mentioned above.  In case of deterministic $A_n$,
Tao and Vu established the existence of the limiting eigenvalue distribution 
independent of the distribution of $h_{ij}$ (see \cite[Theorem 1.23, Theorem 1.7]{TVKr:10}), which means that the limiting distribution is the same for any i.i.d.\ entries under the assumptions $\E\{h_{ij}\}=0$ and $\E\{|h_{ij}|^2\}=1/n$. This result shows that it is enough to consider only the case of Gaussian $H_n$ in order to find the limiting distribution of $A_n+H_n$. The distribution has been studied using free probability theory. \'Sniady \cite{Sn:02} showed that if $A_n$ converges in $*$-moments to some operator $x_0$, then the limiting measure of $X_n$ is equal to the Brown measure of $x_0+c$, where $c$ is Voiculescu's circular operator which is $*$-free from $x_0$. That Brown measure is computed in some cases when $x_0$ belongs to a certain class. The most notable results are obtained in \cite{HoZho} for self-adjoint $x_0$ (which corresponds to the case of hermitian $A_n$) and in  \cite{BoCapCha} for normal $x_0$ with Gaussian spectral measure (which corresponds to to the case of normal $A_n$ with Gaussian limiting distribution of eigenvalues).

In the case of arbitrary $x_0$ (not necessarily normal) the result is obtained by Zhong in his preprint \cite{Zho} by using the free probability techniques (more precisely, free convolutions) and \'Sniady's result \cite{Sn:02}. In terms of matrices, Zhong derives the limiting distribution of $A_n+H_n$ under the only assumption that $A_n$ converges in $*$-moments. 

In this paper, we derive the limiting distribution of $A_n+H_n$ under somewhat different conditions on $A_n$, using an alternative approach.
Our assumptions on $A_n$ appear to be weaker in comparison to the assumptions in \cite{Zho} (see Remark~\ref{r:comp_Zh}). Moreover, apart from weak convergence of the distribution of eigenvalues, we obtain a bound on the rate of convergence of the distribution, which is not presented in~\cite{Zho}.

We use an approach based on supersymmetric integration. Supersymmetry techniques allow us to express the density of normalized counting measure, correlation functions and other spectral characteristics of random matrices as an integral over a set of complex and Grassmann variables. These methods have been successfully applied in various problems in Random Matrix theory, particularly in the study of Gaussian random band matrices (see \cite{EB:15}, \cite{DL:16}, \cite{DPS:02}, \cite{SS:den}, \cite{SS:ChP}, \cite{SS:sig}, \cite{TS:FinB}, \cite{TS:ChP}), for the overlaps of non-Hermitian Ginibre eigenvectors (\cite{Fyo:18}) and for the smallest singular value of Ginibre and deformed Ginibre ensemble (see \cite{SS:sng}, \cite{CiErS:20}).

\subsection{Basic notations and main results}\label{sbsc:bas_notat}

Denote 
\begin{equation}\label{X_n}
X_n=A_n+H_n,
\end{equation}
where $H_n$ is a random $n\times n$ matrix with i.i.d.\ complex Gaussian entries $\{h_{ij}\}_{i,j=1}^n$ satisfying the following conditions:
\begin{equation}\label{Eh}
\E\{h_{ij}\}=0, \quad\E\{h_{ij}^2\}=0,\quad \E\{|h_{ij}|^2\}=1/n,
\end{equation}
and $A_n$ is $n\times n$ matrix with entries $\{a_{ij}\}_{i,j=1}^n$
being either deterministic or random, but independent of~$h_{ij}$. 
Denote
\begin{equation}\label{def_Y_0,Y}
\begin{split}
Y_0(z)=(A_n-z)(A_n-z)^*,\quad Y(z)=(X_n-z)(X_n-z)^*,\\
\wt{Y}_0(z)=(A_n-z)^*(A_n-z),\quad \wt{Y}(z)=(X_n-z)^*(X_n-z).
\end{split}
\end{equation}
Also define normalized trace $\tr$ as $\tr B=\dfrac{1}{n}\Tr B$ for any $n\times n$ matrix $B$.
Our goal is to find the limit of \textit{normalized counting measure} (NCM) of $X_n$. We impose the following conditions on $A_n$:
\begin{itemize}
\item[(C1)] The NCM $\nu_{n,z}$ of $Y_0(z)$ converges weakly to some deterministic measure $\nu_z$ for almost all $z\in\mathbb{C}$. 

\item[(C2)] Denote
\begin{equation*}
\Omega_{M,n}^{(1)}=\{\omega\in\Omega\mid
n^{-1}\sum_{i,j=1}^n |a_{ij}|^2 < M\}.
\end{equation*}
Then there exists some $M>0$ such that $\Prob\{\Omega_{M,n}^{(1)}\}\ge 1-n^{-1-d}$ for some $d>0$. Here and below $\Omega$ stands for the probability space with respect to $A_n$.

\item[(C3)] Denote $\sigma_0=\{z\mid 0\in\supp\nu_{z}\}$. Let $\sigma_\epsilon$ be the $\epsilon$-neighbourhood of $\sigma_0$ and
\begin{equation*}
\begin{split}
&\Omega_{\epsilon,C,n}^{(2)}=\{\omega\in\Omega\mid
\sup_{z\notin\sigma_\epsilon} \Bigl|\tr Y_0^{-1}(z)-\int \lambda^{-1}\,d\nu_z(\lambda)\Bigr| <Cn^{-d_0}\}
\end{split}
\end{equation*}
for some fixed $d_0>0$. Then for some $d>0$ and for all $\epsilon>0$ there exist $C(\epsilon)>0$ satisfying 
\begin{equation*}
\Prob\{ \Omega_{\epsilon,C(\epsilon),n}^{(2)}\} > 1-n^{-1-d}.
\end{equation*}

\item[(C4)] There exist $d_1>0$, $\varrho_0,\epsilon_0>0$ such that if
\begin{equation*}
\Omega_n^{(3)}=\{\omega\in\Omega\mid 
\inf_{z\in\sigma_{\epsilon_0}}\tr (Y_0(z)+\varrho_0^2)^{-1}>1+d_1\}
\end{equation*}
then $\Prob\{\Omega_n^{(3)}\}>1-n^{-1-d}$ for some $d>0$.

\end{itemize}

\begin{rema}
The conditions above are written for the case of random $A_n$. If $A_n$ is deterministic, we assume that the inequalities defining $\Omega^{(j)}$ hold for large $n$. 
\end{rema}

\begin{rema}
Observe that Borel-Cantelli lemma together with (C1)--(C4) implies that 
$$
\Prob\{\exists n_0\colon \omega\in \Omega_{M,n}^{(0)}\cup\Omega_{\epsilon,\kappa,n}^{(1)}\cup\Omega_{\epsilon,C,n}^{(2)}\cup\Omega_n^{(3)}\ \forall n\ge n_0\}=1,
$$ 
allowing us to consider only the case $\omega\in \Omega_{M,n}^{(0)}\cup\Omega_{\epsilon,\kappa,n}^{(1)}\cup\Omega_{\epsilon,C,n}^{(2)}\cup\Omega_n^{(3)}$.
\end{rema}

\begin{rema}
Conditions (C1)--(C4) hold for all classical hermitian ensembles. Here are some other examples of $A_n$ satisfying (C1)--(C4):

\begin{itemize}

\item $A_n$ are diagonal matrices with eigenvalues having limiting distribution with a compact finitely connected support with a smooth boundary such that large deviation type bounds ((C3), (C4)) are satisfied.

\item
$A_n$ is a Ginibre matrix with i.i.d.\ entries having finite fourt moments. In this case bounds of the form (C3), (C4) follow from \cite{AEK:20}.

\end{itemize}

\end{rema}

Define $\mu_{n}$ to be NCM of $X_n$. If $A_n$ is deterministic, then by \cite[Theorem 1.23]{TVKr:10} there exists some deterministic measure $\mu$ which is a weak limit of $\mu_{n}$. If $A_n$ is random, we consider conditioning on the sequence of $A_n$ and conclude that there exists some deterministic with respect to~$H_n$ (but possibly dependent on the sequence of $A_n$) measure $\mu$ which is a weak limit of $\mu_{n}$. The results of the paper include the fact the $\mu$ is in fact fully deterministic.

Consider a region $D\subset\mathbb{C}$ defined as
\begin{equation}\label{def_D}
D=\sigma_0\cup \{z\in\mathbb{C}\setminus \sigma_0\colon \int \lambda^{-1}\,d\nu_{z}(\lambda)\ge 1\}.
\end{equation}
In this paper we show that $D$ is the support of the limiting measure $\mu$. This fact was proved in \cite{BoCap:14} under the additional assumption 
$$
\supp\mu=\{z\mid 0\in\supp\eta_z\}, \text{ where $\eta_z$ is the limit of NCM of $Y(z)=(X_n-z)(X_n-z)^*$}.
$$ 
However, we do not rely on this assumption and prove this result independently.

Observe that (C4) implies $\sigma_{\epsilon_0}\subset D$. Moreover, due to (C3), $\int \lambda^{-1}\,d\nu_{z}(\lambda)=\lim\limits_{n\to\infty} \tr Y_0(z)^{-1}$ is a smooth function in $z$ for $z\in\mathbb{C}\setminus\sigma_{\epsilon_0}$, thus the boundary of $D$ may be found as
\begin{equation}\label{form_dD}
\partial D=\{z\in\mathbb{C}\colon\int \lambda^{-1}\,d\nu_{z}(\lambda)= 1\}
\end{equation}
and consists of several piecewise smooth curves enclosing $\sigma_0$.

Condition (C1) yields that  the limit of 
$$
\mathcal{G}_n(z_1,z_2,x)=n^{-1}\log\det\bigl((z_1+iz_2-A_n)(z_1-iz_2-A_n^*)+x\bigr)
$$ 
as $n\to\infty$ is equal to the non-random function $\int \log(\lambda+x)\,d\nu_{z_1+iz_2}(\lambda)$ for real $z_1,z_2$ and $x>0$. Since $\mathcal{G}_n(z_1,z_2,x)$ is analytic in $z_1,z_2,x$, one can find a non-random analytic continuation $\mathcal{G}(z_1,z_2,x)$ of $\int \log(\lambda+x)\,d\nu_{z_1+iz_2}(\lambda)$ such that all derivatives of $\mathcal{G}_n(z_1,z_2,x)$ converge to the respective derivatives of $\mathcal{G}(z_1,z_2,x)$. Thus, if we consider
\begin{equation}\label{T_1,T_2}
\begin{split}
T_1(z,x)=\frac{1}{2}\bigl(\partial_{z_1}+i\partial_{z_2}\bigr)\partial_x \mathcal{G}(z_1,z_2,x)\,\Big|_{z_1=\Re z,\  z_2=\Im z,}\\
T_2(z,x)=\frac{1}{4x}\bigl(\partial_{z_1}^2+\partial_{z_2}^2)\mathcal{G}(z_1,z_2,x)\,\Big|_{z_1=\Re z,\  z_2=\Im z,}
\end{split}
\end{equation}
then $T_1(z,x)=\lim\limits_{n\to\infty} \tr (A_n-z)(Y_0(z)+x)^{-2}$ and $T_2(z,x)=\lim\limits_{n\to\infty} \tr (Y_0(z)+x)^{-1}(\wt{Y}_0(z)+x)^{-1}$ for all $x>0$, and $T_{1,2}$ are non-random.

We are ready to formulate the main result of the paper.
\begin{thm}\label{th:main}
Assume that $X_n$ defined in \eqref{X_n} satisfies \eqref{Eh} and conditions (C1)--(C4).
Set 
\begin{equation}\label{rho_main_form}
\rho_\mu(z)=
\dfrac{1}{\pi}\Bigl( 
\dfrac{|T_1(z,x_{0}^2)|^2}{\int (\lambda+x_{0}^2)^{-2}\,d\nu_z(\lambda)}+
x^2_{0}\cdot T_2(z,x_0^2)
 \Bigr)\cdot \mathbf{1}_D(z),
\end{equation}
where $x_0=x_0(z)>0$ satisfies the equation $\int (\lambda+x_0^2)^{-1}\,d\nu_z(\lambda)=1$, $T_1,T_2$ are defined in \eqref{T_1,T_2} and $\mathbf{1}_D$ is the characteristic function of $D$ defined in~\eqref{def_D}.
Then, for any  $h(z)\in C_c^2(\mathbb{C})$ we have 
$$
\int h(z)\,d\mu(z)=\int h(z)\rho_\mu(z)\,\dz.
$$
\end{thm}

The above result yields that $\mu$ does not actually depend on the particular choice of $A_n$. In other words, there is a fully deterministic measure $\mu$ with density $\rho_\mu(z)$ given by \eqref{rho_main_form} which is a limit of NCM of $X_n$. More precisely,

\begin{cor}
Assume that $X_n$ defined in \eqref{X_n} satisfies \eqref{Eh} and conditions (C1)--(C4).
Then the normalized counting measure of $X_n$ converges weakly to the measure $\mu$ with density $\rho_\mu(z)$ defined in~\eqref{rho_main_form}.
\end{cor}

\begin{rema}\label{r:comp_Zh}
As already mentioned,  the same result is obtained in \cite{Zho} under the assumption that $A_n$ converges in $*$-moments almost surely to a non-commutative random variable $x_0$ (which with some simplification means that $\tr A_n^{e_1}A_n^{e_2}\ldots A_n^{e_k}$ converges almost surely as $n\to\infty$ for any $k$ and any $e_1,\ldots,e_k\in\{1;*\}$). Notice that our main condition (C1) on the convergence of $A_n$ is weaker than convergence in $*$-moments. Moreover, convergence in $*$-moments requires that the entries of $A_n$ have all moments, and such condition appears somewhat unnatural. In our case, the existence of moments is not required for (C1), while conditions (C2) -- (C4) hold for all classical ensembles with entries having finite fourth moment. 
\end{rema}

We also derive a bound on the rate of weak convergence of NCM. Naturally, we need some additional condition on the convergence rate of $A_n$, which is formulated in terms of NCM $\nu_{n,z}$ of $Y_0(z)$:

\begin{itemize}
\item[(C5)] For $\kappa>0$, $C>0$ denote
\begin{equation*}
\begin{split}
&\Omega_{\kappa,C,n}^{(4)}=\{\omega\in\Omega\mid
\sup_{z\in D,\, \kappa\le x\le 2} \Bigl|n^{-1}\log\det (Y_0(z)+x)-\int \log(\lambda+x)\,d\nu_z(\lambda)\Bigr| <Cn^{-1}\}.
\end{split}
\end{equation*}
Then for some $d>0$ and for all $\kappa>0$ there exists $C(\kappa)>0$ satisfying 
\begin{equation*}
\Prob\Bigl\{\Omega_{\kappa,C(\kappa),n}^{(4)}\Bigr\} > 1-Cn^{-1-d}.
\end{equation*}

\end{itemize}

We establish the following result:

\begin{thm}\label{thm:rate_of_conv}
Assume that $X_n$ defined in \eqref{X_n} satisfies \eqref{Eh} and conditions (C1)--(C5). Let $z_1,\ldots,z_n$ be the eigenvalues of $X_n$ and $h(z)\in C_c^2(D)$. Then
\begin{equation*}
\Bigl|\E\Bigl\{\dfrac{1}{n}\sum_{j=1}^n h(z_j)\Bigr\}-\int h(z)\,d\mu(z)\Bigr|\le Cn^{-1/2}.
\end{equation*}
\end{thm}

\begin{rema}
One can see from the computations given in Section~\ref{sc:uni_conv_prf} and Section~\ref{sc:rate} that a similar result holds for an arbitrary smooth $h(z)$ with compact support but with the error term $O(n^{-\alpha})$ for $\alpha$ much smaller than $\dfrac{1}{2}$.  

\end{rema}

The paper is organized as follows. In Section~2 the method of computing the limiting density is described. In Sections~3--6 we perform a step-by-step realization of that method (see the end of Section~2 for more details). In Section~7 the rate of convergence is discussed.

\subsubsection*{Acknowledgement}

The author is grateful to Prof. M. Shcherbina for the statement of the problem and productive discussions. This work is supported by Grant ``International Multilateral Partnerships for Resilient Education and Science System in Ukraine'' IMPRESS-U: N7114
funded by US National Academy of Science and Office of Naval Research.

\section{Strategy for computation of the limiting density }

Let $z_1,\ldots,z_n$ be the eigenvalues of $X_n$. According to the standard potential theory, 
$$
\rho_{n}(z)=\frac{1}{4\pi n}\Delta\log\det Y(z)=\frac{1}{4\pi n}\Delta\sum\limits_{j=1}^n\log|z-z_j|^2
$$ is the density of NCM $\mu_{n}$ of $X_n$ in the sense of distribution, i.e.
\begin{equation}\label{potent_formula}
\int \Delta h(z)\cdot\frac{1}{4\pi n}\log\det Y(z)\,\dz=\frac{1}{n}\sum_{j=1}^n h(z_j)=\int h(z)\,d\mu_{n}
\end{equation}
for an arbitrary $h(z)\in C_c^2(\mathbb{C})$, where $Y(z)$ is defined in \eqref{def_Y_0,Y}, $\Delta$ is a two-dimensional Laplacian on $\mathbb{C}$ and $\dz$ is a standard two-dimensional Lebesgue measure on $\mathbb{C}$. In this paper the limiting density is found by considering some `regularization' of the density $\rho_{n}(z)$ and then taking the limit as $n\to\infty$. 

One can consider a regularization $\wt{\rho}_{\ep,n}(z)=\frac{1}{4\pi n}\sum\limits_{j=1}^n\Delta\log(|z-z_j|^2+\ep^2)$ and show that 
$$
\lim\limits_{\ep\to 0}\int h(z)\wt{\rho}_{\ep,n}(z)\,\dz=\int h(z)\rho_{n}(z)\,\dz
$$
uniformly in $n$, which allows us to find the density of the limiting measure $\mu$ as the double limit $\lim\limits_{\substack{\ep\to 0\\ n\to\infty}} \wt{\rho}_{\ep,n}$. 
However, 
$$
\sum\limits_{j=1}^n\log(|z-z_j|^2+\ep^2)\neq\log\det(Y(z)+\ep^2)
$$ 
in general for non-hermitian matrices $X_n$, 
and there is no easy way to express $\wt{\rho}_{\ep,n}$ in terms of $X_n$. This fact shows that the regularization $\wt{\rho}_{\ep,n}$ is not suitable for our problem.
Instead, we 
use the following regularization: 
\begin{equation}\label{rho_ep,n_first_def}
\rho_{\ep,n}(z)=\frac{1}{4\pi n}\Delta\,\log\det (Y(z)+\ep^2).
\end{equation}
Suppose that this regularization is somehow computed. In order to find the density of the limiting measure $\mu$, we need to somehow send $\ep\to 0$ and $n\to\infty$ in $\rho_{\ep,n}$. From the definition of the regularization one can see that $\lim\limits_{\ep\to 0} \rho_{\ep,n}=\rho_{n}$ and $\lim\limits_{n\to \infty} \rho_{n}=\rho_{\mu}$ in the sense of distributions, where $\rho_{\mu}$ is the density of the limiting $\mu$. Thus it would be natural to try computing $\rho_{\mu}$ as iterated limit $\lim\limits_{n\to\infty}\lim\limits_{\ep\to 0} \rho_{\ep,n}$. However, in practice we are able to derive the asymptotic behaviour of $\rho_{\ep,n}$ for fixed $\ep>0$ as $n\to\infty$. Hence it is convenient for us to compute another iterated limit $\lim\limits_{\ep\to 0} \lim\limits_{n\to\infty}\rho_{\ep,n}$, but then we need to prove that these two iterated limits are equal. Due to Moore-Osgood theorem it is enough to check that the distribution-wise convergence $\lim\limits_{\ep\to 0} \rho_{\ep,n}(z)=\rho_{n}(z)$ is uniform in $n$, and the proof of such fact appears to be the main technical difficulty of this paper.

Now let us be more precise. Denote  $\E_{H_n}$ the expectation with respect to the entries of $H_n$ (i.e. conditional expectation given $A_n$), and set
\begin{equation}\label{rho_ep,n_def}
\bar{\rho}_{\ep,n}(z)=\E_{H_n}\{\rho_{\ep,n}(z)\}
=\frac{1}{4\pi n}\Delta\E_{H_n}\Bigl\{\log\det (Y(z)+\ep^2)\Bigr\}.
\end{equation}
The following result is proven in Section \ref{sc:uni_conv_prf} using the integral representation of $\partial_\ep\E_{H_n}\{\log\det (Y(z)+\ep^2)\}$:

\begin{pro}\label{l:uni_conv}
Assume that $X_n$ defined in \eqref{X_n} satisfies \eqref{Eh} and conditions (C1)--(C4). Let $z_1,\ldots,z_n$ be the eigenvalues of $X_n$ and $h(z)\in C_c^2(\mathbb{C})$. Then there exists $n_0\in\mathbb{N}$ such that
$$
\lim\limits_{\ep\to 0}\ \int h(z) \bar{\rho}_{\ep,n}(z)\,\dz= \E_{H_n}\Bigl\{\frac{1}{n}\sum\limits_{j=1}^n h(z_j)\Bigr\}
$$
uniformly in~$n$ for $n\ge n_0$. In other words, the distribution-wise convergence
$$
\lim\limits_{\ep\to 0} \E_{H_n}\{\rho_{\ep,n}\}=\E_{H_n}\{\rho_{n}\}
$$
is uniform in $n$ for $n\ge n_0$.
\end{pro}
As mentioned above, this fact allows us to change the order of the limits in $\lim\limits_{\ep\to 0} \lim\limits_{n\to\infty}\rho_{\ep,n}$. More precisely,

\begin{cor}\label{pro:alg_lim_meas}
Take an arbitrary $h(z)\in C_c^2(\mathbb{C})$ and denote $E=\supp h$.
Suppose that Proposition~\ref{l:uni_conv} holds for $\bar{\rho}_{\ep,n}(z)$. Assume that for almost all $z\in E$ there exists $\lim\limits_{n\to\infty}\bar{\rho}_{\ep,n}(z)=\rho_{\ep}(z)$ such that $|\rho_{\ep}(z)|\le C$ for all $0<\ep\le \ep_0$ and $z\in E$, and there exists $\lim\limits_{\ep\to 0} \rho_{\ep}(z)=\rho(z)$ for all $z\in E$.  Then 
$$
\int h(z) \rho(z)\,\dz = \int h(z) d\mu(z).
$$
In other words, $\rho=\lim\limits_{\ep\to 0} \lim\limits_{n\to\infty}\bar{\rho}_{\ep,n}$ is indeed the density of $\mu$.
\end{cor}

\begin{proof}

Again, let $\{z_1,z_2,\ldots,z_n\}$ be the set of eigenvalues of $X_n$, and recall that $\mu$ defined in Subsection~\ref{sbsc:bas_notat} is the weak limit of $\mu_{n}$. 
Then dominated convergence theorem implies
\begin{equation*}
\lim_{n\to\infty}\E_{H_n}\Bigl\{\dfrac{1}{n}\sum\limits_{j=1}^n h(z_j) \Bigr\} = \int h(z)\, d\mu.
\end{equation*}
One can check that
$
\bar{\rho}_{\ep,n}(z)=\dfrac{1}{\pi}\ep^2
\E_{H_n}\Bigl\{\tr (Y(z)+\ep^2)^{-1}(\wt{Y}(z)+\ep^2)^{-1}\Bigr\}\le \dfrac{1}{\pi\ep^2}.
$
This bound, together with the assumptions of the theorem, gives us
\begin{equation*}
\begin{split}
&\lim_{n\to \infty}\ \int h(z) \bar{\rho}_{\ep,n}(z)\,\dz=
\int h(z) \rho_{\ep}(z)\,\dz,\\
&\lim_{\ep\to 0}\ \int h(z) \rho_{\ep}(z)\,\dz=
\int h(z) \rho(z)\,\dz.
\end{split}
\end{equation*}
All this convergences combined with uniform convergence from Proposition~\ref{l:uni_conv} and Moore-Osgood theorem imply the identity $\int h(z) \rho(z)\,\dz=\int h(z) \,d\mu$. 

\end{proof}

\begin{rema}\label{rm:to_prove}
Comparing Corollary~\ref{pro:alg_lim_meas} with Theorem \ref{th:main} one can see that it is sufficient to check the following facts:
\begin{itemize}
\item There exist limits $\lim\limits_{n\to\infty}\bar{\rho}_{\ep,n}(z)=\rho_{\ep}(z)$ and $\lim\limits_{\ep\to 0} \rho_{\ep}(z)=\rho(z)$ for almost all $z\in \mathbb{C}$;
\item $\rho_{\ep}(z)$ is bounded uniformly in $\ep\le\ep_0$ and almost all $z\colon |z|\le C$;
\item $\rho(z)=\rho_\mu(z)$ where $\rho_\mu(z)$ is given by \eqref{rho_main_form};
\item Proposition \ref{l:uni_conv} holds.
\end{itemize}
\end{rema}
\medskip

\noindent
The rest of the paper is organized as follows. In Section 3 the integral representation of $\bar{\rho}_{\ep,n}(z)$ is obtained. In Section 4 the derived integral representation and saddle point method are used to study asymptotic behaviour of $\bar{\rho}_{\ep,n}(z)$ as $n\to\infty$. In Section 5 the limits $\rho_{\ep}(z)=\lim\limits_{n\to\infty} \bar{\rho}_{\ep,n}(z)$ and $\rho(z)=\lim\limits_{\ep\to 0} \rho_{\ep}(z)$ are found. Section 6 contains the proof of Proposition~\ref{l:uni_conv}. Section 7 is devoted to the rate of convergence of NCM $\mu_{n}$.

\section{Integral representation of \boldmath{$\bar{\rho}_{\ep,n}(z)$}}

Below the formula for $\bar{\rho}_{\ep, n}$ is rewritten to be more suitable for supersymmetric integration. We use the trick introduced by Fyodorov, Sommers in \cite{FyoSom:03}. Recall that $\Delta=4\partial_{\bar z}\partial_z$, where $\partial_z$ and $\partial_{\bar{z}}$ are Wirtinger derivatives defined as $\partial_z f(z)=\frac{1}{2}(\partial_x-i\partial_y)f(x+iy)$,  $\partial_{\bar z} f(z)=\frac{1}{2}(\partial_x+i\partial_y)f(x+iy)$. Straightforward computation shows that
\begin{equation}\label{rho_ep,n}
\partial_{\bar{z}}\partial_z\log\det (Y(z)+\ep^2)=
\partial_{\bar{z}}\left(\partial_{z_1}\dfrac{\det(Y(z_1)+\ep^2)}{\det(Y(z)+\ep^2)}\right)\Bigr|_{z_1=z}.
\end{equation}
This identity together with \eqref{rho_ep,n_def} implies
\begin{equation}\label{rho_en_form}
\bar{\rho}_{\ep, n}(z)=\dfrac{1}{\pi n}
\partial_{\bar{z}}\Bigl(\bigl(\partial_{z_1}\mathcal{Z}(\ep,\ep,z,z_1)\bigr)\bigr|_{z_1=z}\Bigr),
\end{equation}
where 
\begin{equation}\label{def_Z}
\mathcal{Z}(\ep,\ep_1,z,z_1)=\E_{H_n}\left\{\dfrac{\det(Y(z_1)+\ep_1^2)}{\det(Y(z)+\ep^2)}\right\}.
\end{equation}

The following proposition gives us an integral representation of $\mathcal{Z}(\ep,\ep_1,z,z_1)$.

\begin{pro}\label{pr:I_rep_Z}
We have 
\begin{equation}\label{I_rep_Z}
\begin{split}
\mathcal{Z}(\ep,\ep_1,z,z_1)
=&\dfrac{2n^3}{\pi^3}\int_0^\infty dR\int_{-\infty}^\infty dv\,du_1\,du_2\,ds\int_L dt\cdot \dfrac{R}{\sqrt{v^2+4R}}\cdot \varphi(u_1^2+u_2^2,s^2-t^2,z,z_1)\times\\
&\times  \exp\left\{n\Bigl(\L_n(z_1,u_1^2+u_2^2)-(u_1+\ep_1)^2-u_2^2\Bigr)\right\}\times\\
&\times  \exp\left\{-n\Bigl(\L_n(z,s^2-t^2)+(t-i\ep)^2+(R+it+\ep)^2+\ep v^2\Bigr)\right\},
\end{split}
\end{equation}
with  $L:=\mathbb{R}+\ep_0i$,\quad $\varphi(x,y,z,z_1)=\varphi_{1}(x,y,z,z_1)-\dfrac{1}{n}\varphi_{2}(x,y,z,z_1)$,
\begin{equation}\label{notat_I_rep_Z}
\begin{split}
&\begin{split}
\varphi_{1}(x,y,z,z_1)=&\left(1-\tr (A_n-z_1)^*G(z_1,x)G(z,y)(A_n-z)\right)\times\\
&\times\left(1-\tr G(z_1,x)(A_n-z_1)(A_n-z)^*G(z,y)\right)-\\
&-xy\cdot \tr G(z_1,x)G(z,y)
\cdot\tr \wt{G}(z_1,x)\wt{G}(z,y);
\end{split}\\
&\begin{split}
\varphi_{2}(x,y,z,z_1)=&
y\cdot\tr G(z_1,x)(A_n-z_1)\wt{G}(z,y)(A_n-z_1)^*G(z_1,x)G(z,y)+\\
&+x\cdot\tr G(z_1,x)G(z,y)(A_n-z)\wt{G}(z_1,x)(A_n-z)^*G(z,y),
\end{split}
\end{split}
\end{equation}
where 
\begin{equation}\label{def_L_G}
\begin{split}
&\L_n(z, x)=\frac{1}{n}\log\det\bigl(Y_0(z)+x\bigr)=\frac{1}{n}\log\det\bigl((z-A_n)(\bar{z}-A_n^*)+x\bigr),\\
&G(z,x)=(Y_0(z)+x)^{-1}=((z-A_n)(\bar z-A_n^*)+x)^{-1},\\
&\wt{G}(z,x)=(\wt{Y}_0(z)+x)^{-1}=((\bar z-A_n^*)(z-A_n)+x)^{-1}.
\end{split}
\end{equation}

\end{pro}

This integral representation was in fact established in \cite{SS:sng} (see the proof of Proposition~2.1). They obtained the integral representation of $\mathcal{Z}(\ep,\ep_1,z,z)$ (which is $\mathcal{Z}(\ep,\ep_1)$ in \cite{SS:sng}) and then differentiated the identity with respect to $\ep_1$ in order to get $T(z,\ep)$. One can obtain the representation of $\mathcal{Z}(\ep,\ep_1,z,z_1)$ using the same strategy, changing $z$ to $z_1$ in the part which corresponds to Grassmann variables.
After that just make the following changes of variables:
\begin{equation*}
\begin{split}
(t_1,t_2)\to (s,t),\quad s=\dfrac{t_1-t_2}{2},\quad t=\dfrac{t_1+t_2}{2},\quad s\in\mathbb{R},\ t\in L,\\
(r_1,r_2)\to (v,R), \quad v=r_1-r_2,\quad R=r_1r_2,\quad v\in\mathbb{R},\ R\in[0,+\infty).
\end{split}
\end{equation*}
The Jacobians of these changes are $J_1=2$ and $J_2=\dfrac{1}{(v^2+4R)^{1/2}}$. It is easy to see that we obtain \eqref{I_rep_Z}.

\begin{rema}
We need to know $\mathcal{Z}(\ep,\ep_1,z,z_1)$ only for $\ep_1=\ep$ in order to find $\bar{\rho}_{\ep,n}(z)$. The formula for $\ep\neq \ep_1$ is used in Section 6 to prove Proposition~\ref{l:uni_conv}.
\end{rema}

\begin{rema}\label{rm:phi_z_1=z}
It follows from the proof of \cite[Proposition 2.1]{SS:sng} that for $z_1=z$ we have
\begin{equation*}
\begin{split}
&\varphi_{1}(x,y,z,z)=\Big(1-\tr G(z,y)+x\,\tr G(z,x)G(z,y)\Big)^2-xy(\tr G(z,x)G(z,y))^2;\\
&\varphi_{2}(x,y,z,z)=(x-y)(\tr G(z,x)G^2(z,y)-x\,\tr G^2(z,x)G^2(z,y)).
\end{split}
\end{equation*}
\end{rema}

In order to get an integral representation of $\bar{\rho}_{\ep,n}$, we need to differentiate \eqref{I_rep_Z} with respect to $z, z_1$ as in \eqref{rho_en_form}. To this end, we need to set $\ep_1=\ep$ and isolate the parts of the integrand which depend on~$z,z_1$. Introduce a functional 
\begin{equation}\label{def_averaging}
\begin{split}
\mathcal{I}(f(u_1,u_2,t,s,z,z_1))
=&\dfrac{2n^3}{\pi^3}\int_0^\infty dR\int_{-\infty}^\infty dv\,du_1\,du_2\,ds\int_L dt\cdot \dfrac{R}{\sqrt{v^2+4R}}\cdot f(u_1,u_2,t,s,z,z_1)\times\\
&\times \exp\{-n((u_1+\ep)^2+u_2^2+(t-i\ep)^2+(R+it+\ep)^2+\ep v^2)\},
\end{split}
\end{equation}
We can use it to rewrite \eqref{I_rep_Z} as
\begin{equation}\label{Z_repr_aver}
\mathcal{Z}(\ep,\ep,z,z_1)=\mathcal{I}(\varphi(z,z_1) e^{n\F_n(z,z_1)}),
\end{equation}
where $\mathcal{F}_n(z,z_1)=\L_n(z_1,u_1^2+u_2^2)-\L_n(z,s^2-t^2)$ and $\varphi(z,z_1)$ is a short notation for the function $\varphi(u_1^2+u_2^2,s^2-t^2,z,z_1)$ defined in \eqref{notat_I_rep_Z}.
Further we will use the following trivial observation:
\begin{equation}\label{<phi_exp_nF>=1}
\mathcal{I}( \varphi(z,z) e^{n\F_n(z,z)})=\mathcal{Z}(\ep,\ep,z,z)=1.
\end{equation}
Now notice that
\begin{equation*}
\partial_{\bar{z}}\bigl((\partial_{z_1}\mathcal{Z}(\ep,\ep,z,z_1))\bigr|_{z_1=z}\bigr)=
\bigl((\partial_{\bar{z}}+\partial_{\bar{z}_1})\partial_{z_1}\mathcal{Z}(\ep,\ep,z,z_1)\bigr)\bigr|_{z_1=z}
\end{equation*}
The last identity combined with \eqref{rho_en_form} and \eqref{Z_repr_aver} implies the next result:

\begin{pro}\label{pro:int_rep_dens}
Consider the following integrals: 
\begin{equation}\label{def_I_k}
\begin{split}
&\I_1:=\mathcal{I}\bigl(\partial_{z_1}\F_n(z,z_1)\cdot(\partial_{\bar{z}}+\partial_{\bar{z}_1})\bigl( \varphi(z,z_1)\,e^{n\F_n(z,z_1)}\bigr)\bigr)\Bigr|_{z_1=z};\\ 
&\I_2:=\mathcal{I}\bigl((\partial_{\bar{z}}+\partial_{\bar{z}_1}) \partial_{z_1}\F_n(z,z_1)\cdot \varphi(z,z_1)\,e^{n\F_n(z,z_1)}\bigr)\Bigr|_{z_1=z}; \\
&\I_3:=\mathcal{I}\bigl( \partial_{z_1}\varphi(z,z_1)\cdot(\partial_{\bar{z}}+\partial_{\bar{z}_1})\F_n(z,z_1)\cdot \,e^{n\F_n(z,z_1)}\bigr)\Bigr|_{z_1=z};\\
&\I_4:=\dfrac{1}{n}\mathcal{I}\bigl((\partial_{\bar{z}}+\partial_{\bar{z}_1})\partial_{z_1}\varphi(z,z_1)\cdot e^{n\F_n(z,z_1)}\bigr)\Bigr|_{z_1=z}.
\end{split}
\end{equation}
Then 
$
\bar{\rho}_{\ep,n}(z)=\dfrac{1}{\pi}(\I_1+\I_2+\I_3+\I_4).
$
\end{pro}

\section{Asymptotic behaviour of \boldmath{$\bar{\rho}_{\ep,n}(z)$}}

In this section we perform an asymptotic analysis of $\I_1,\I_2,\I_3,\I_4$ as $n\to\infty$ for some fixed $\ep>0$ and $z\in\mathbb{C}\setminus \partial D$, which will give us the asymptotic behaviour of $\bar{\rho}_{\ep,n}(z)$ as $n\to\infty$.

Below we will write $\L_n(x)$, $G(x)$ and $\wt{G}(x)$ instead of $\L_n(z,x)$, $G(z,x)$ and $\wt{G}(z,x)$ to simplify the notations.

\subsection{Preparations for the saddle point method}\label{sbsc:saddle}

We will use the saddle point method to analyse certain integrals. First, we study the solutions of some equations that will appear further as saddle points. Assume that conditions (C1)--(C4) hold, and consider the following equations:
\begin{align}
&1-\tr G(x^2)=\dfrac{\ep}{x},\label{eq_x_en}\\
&1-\int (\lambda+x^2)^{-1} \,d\nu_z(\lambda)=\dfrac{\ep}{x},\label{eq_x_e}\\
&\tr G(x^2)=1,\label{eq_x_0n}\\
&\int (\lambda+x^2)^{-1} \,d\nu_z(\lambda)=1,\label{eq_x_0}
\end{align}
where $\ep>0$. 

Let us fix a compact set $E_{in}$ satisfying $E_{in}\subset\Int D$. First we study the solutions of the equations above for $z\in E_{in}$. 

\begin{pro}\label{pro:eq_sols_int}
The equations \eqref{eq_x_e} and \eqref{eq_x_0} have exactly one positive root each for $z\in\Int D$. Moreover, there exists $n_0=n_0(E_{in})$ such that for $n\ge n_0$ and $z\in E_{in}$ each of the equations \eqref{eq_x_en} and \eqref{eq_x_0n} has exactly one positive root.
\end{pro}

Denote $x_{\ep,n}$, $x_{\ep}$, $x_{0,n}$ and $x_0$ the positive solutions of \eqref{eq_x_en}, \eqref{eq_x_e}, \eqref{eq_x_0n} and \eqref{eq_x_0} correspondingly.

\begin{pro}\label{pro:sols_bound_int}
For a given $E_{in}$, there exist $\kappa_0=\kappa_0(E_{in})>0$, $n_1=n_1(E_{in})$ and $\ep_0=\ep_0(E_{in})$ such that for all $z\in E_{in}$, $n\ge n_1$ and $\ep\le \ep_0$ the following inequalities hold:\smallskip

1. $\kappa_0\le x_{0,n}\le 1$,\qquad
2. $x_{0,n}\le x_{\ep,n}\le x_{0,n}+\dfrac{\ep}{\kappa_0^2}$.

\noindent
Also we have \quad $\lim\limits_{\ep\to 0} x_{\ep,n}= x_{0,n}$,
$\lim\limits_{n\to \infty} x_{0,n}= x_{0}$,
$\lim\limits_{n\to \infty} x_{\ep,n}= x_{\ep}$,
$\lim\limits_{\ep\to 0} x_{\ep}= x_{0}$.

\end{pro}

Now let us fix a compact set $E_{out}$ satisfying $E_{out}\subset\mathbb{C}\setminus \ol{D}$. We study the solutions of the equations above for $z\in E_{out}$. 

\begin{pro}\label{pro:eq_sols_ext}
The equations  \eqref{eq_x_en} and \eqref{eq_x_e} have exactly one positive root each, while \eqref{eq_x_0} has no nonnegative roots for $z\in\mathbb{C}\setminus\ol{D}$. Moreover, there exists $n_0=n_0(E_{out})$ such that for $n\ge n_0$ and $z\in E_{out}$ the equation \eqref{eq_x_0n} has no nonnegative roots.
\end{pro}

As before, we denote $x_{\ep,n}$, $x_{\ep}$ the positive solutions of \eqref{eq_x_en}, \eqref{eq_x_e} correspondingly.

\begin{pro}\label{pro:sols_bound_ext}
For a given $E_{out}$, there exist $K_0=K_0(E_{out})>0, \kappa_0=\kappa_0(E_{out})>0$ and $n_1=n_1(E_{out})$ such that for all $z\in E_{out}$ and $n\ge n_1$ the following inequality holds:
$
\ep(1+\kappa_0)\le x_{\ep,n}\le \ep(1+K_0).
$
Also, we have $\lim\limits_{n\to\infty} x_{\ep,n}=x_\ep$.
\end{pro}

One can easily show that Propositions~\ref{pro:eq_sols_int}--\ref{pro:sols_bound_ext} follow from conditions (C1)--(C4) and Rouch\'e's theorem.

Now consider a set $E_{\partial D}=\{z\in\mathbb{C}\mid \dist(z,\partial D)\le d\}$ for some small $d$. We have the following result:

\begin{pro}\label{pro:sols_bound_boundary}
For any $z\in E_{\partial D}$ there exists exactly one positive solution $x_{\ep,n}$ of \eqref{eq_x_en} and exactly one positive solution $x_{\ep}$ of \eqref{eq_x_e}. Moreover, one can find $C=C(E_{\partial D})>0$ and $c=c(E_{\partial D})>0$ such that
\begin{equation*}
\begin{split}
&c\ep^{1/3}\le x_{\ep,n}\le C, \quad {when\ } \tr G(0)\ge 1;\quad
(1+c)\ep\le x_{\ep,n}\le C\ep^{1/3},\quad {when\ } \tr G(0)<1;\\
& c\ep^{1/3}\le x_{\ep}\le C, \quad {when\ } z\in D\cap E_{\partial D};\quad
(1+c)\ep\le x_{\ep}\le C\ep^{1/3},\quad {when\ } z\in E_{\partial D}\setminus D.
\end{split}
\end{equation*}
\end{pro}

\begin{proof}
Suppose that $\tr G(0)\ge 1$. The upper bound on $x_{\ep,n}$ is obvious. Next, we have
\begin{equation*}
\begin{split}
\dfrac{\ep}{x_{\ep,n}}=1-\tr G(x_{\ep,n}^2)\le \tr G(0)-\tr G(x_{\ep,n}^2)=\tr G^2(\xi)\cdot x_{\ep,n}^2\le C x_{\ep,n}^2
\end{split}
\end{equation*}
for some $\xi\in(0,x_{\ep,n}^2)$, which gives us the lower bound.

Now suppose that $\tr G(0)< 1$. For the upper bound, observe that
\begin{equation*}
\begin{split}
\dfrac{\ep}{x_{\ep,n}}=1-\tr G(x_{\ep,n}^2)> \tr G(0)-\tr G(x_{\ep,n}^2)=\tr G^2(\xi)\cdot x_{\ep,n}^2> c x_{\ep,n}^2.
\end{split}
\end{equation*}
For the lower bound, we have
$
\dfrac{\ep}{x_{\ep,n}}=1-\tr G(x_{\ep,n}^2)< 1-\tr G(C^2)<1-\kappa
$
for some $\kappa>0$, since $x_{\ep,n}<C\ep^{1/3}<C$. 

One can similarly obtain the bounds on $x_{\ep}$ using $\int (\lambda+x)^{-1}\,d\nu_z(\lambda)$ instead of $\tr G(x)$.
\end{proof}

\subsection{Integrals of the form $\mathcal{I}\bigl( g \cdot e^{n\F_n(z,z)}\bigr)$}\label{sbsc:<ge^nF>}

In order to derive the asymptotic behaviour of $\I_k$, 
the integrals of more general form are studied. Set
\begin{equation}\label{def_I_g_n}
\I_{g_n}=\mathcal{I}\bigl( g_n(u_1^2+u_2^2,s^2-t^2) e^{n\F_n(z,z)}\bigr),
\end{equation}
where a sequence of complex-valued functions $g_n(x,y)$ satisfies the following conditions:
\begin{enumerate}
\item $g_n(x,y)$ are analytic in some neighbourhood of $(x_{\ep,n}^2,x_{\ep,n}^2)$;
\item $g_n(x,y)$ are bounded uniformly in~$n$ in some neighbourhood of $(x_{\ep,n}^2,x_{\ep,n}^2)$;
\item $\mathcal{I}\bigl( |g_n(u_1^2+u_2^2,s^2-t^2) e^{N_1 \F_n(z,z)}|\bigr)\le C$ for some fixed $N_1$ and $C>0$.
\end{enumerate}
Our goal is to prove that 
\begin{equation}\label{form_<ge^nF>}
\I_{g_n}=\dfrac{g_n(u_1^2+u_2^2,s^2-t^2)}{\varphi(u_1^2+u_2^2,s^2-t^2,z,z)}\Bigg|_{\substack{u_1=-x_{\ep,n},\ u_2=0,\\ t=ix_{\ep,n},\ s=0}}+O(n^{-\alpha}),
\end{equation}
where $x_{\ep,n}$ is the positive root of \eqref{eq_x_en} and $\varphi$ is defined in \eqref{notat_I_rep_Z}.

In this subsection we fix some $\ep >0$ and some $z\in\partial D$. In case $z\in\Int D$ we set $E_{in}:=\{z\}$, while in case $z\in\mathbb{C}\setminus \ol{D}$ we set $E_{out}:=\{z\}$. Then we can use the results of Subsection~\ref{sbsc:saddle}. In particular, a bound of the form 
\begin{equation}\label{bound_x_en}
c_{\ep}\le x_{\ep,n}\le C_{\ep}
\end{equation}
holds, where $c_\ep, C_\ep>0$ are independent of $n$, but may depend on $\ep$ which is now fixed.  Henceforth in this section, the multiplicative constant in expressions of the form $O(f(n))$ and constants denoted by $c,C$ may depend on $\ep$.

Recall the definition \eqref{def_averaging}. Make a change of variable $r:=R+it+\ep$ and denote 
$$
\mathcal{R}(t)=\{r\colon r=it+\ep+\rho, \rho\ge 0\}.
$$ 
Notice that the integrand is even with respect to $u_2$, thus we can integrate with respect to $u_2$ over $[0,+\infty)$ and write a multiplier 2 before the integral. Also we have
$$
L_n(u_1^2+u_2^2)-(u_1+\ep)^2-u_2^2=L_n(u_1^2+u_2^2)-\Bigl(\sqrt{u_1^2+u_2^2}-\ep\Bigl)^2-2\ep\,\Bigl(u_1+\sqrt{u_1^2+u_2^2}\,\Bigr).
$$
Since $u_1+\sqrt{u_1^2+u_2^2}\ge 0$, we can make a change of variables $(u_1,u_2)\to (u,w)$, $u=\sqrt{u_1^2+u_2^2}\in[0,+\infty)$, $w=\sqrt{u_1+u}\in[0,\sqrt{2u}]$. One can see that $u_1=w^2-u$, $u_2=\sqrt{u^2-u_1^2}$ and the Jacobian of this change is equal to $J=-\dfrac{2u}{\sqrt{2u-w^2}}$.

Set $\mathbf{x}=(u,t,s,r,v,w)$ and $\mathcal{V}=[0,+\infty)\times L_t\times \mathbb{R}\times \mathcal{R}(t)\times \mathbb{R}\times[0,\sqrt{2u}]$. After all of the changes above, we obtain 
\begin{equation*}
\I_{g_n}=\frac{8n^3}{\pi^3}\int_{\mathcal{V}} \Phi_n(\mathbf{x})\,e^{nF_n(\mathbf{x})}\,d\mathbf{x},
\end{equation*}
where
\begin{equation*}
\begin{split}
&F_{n,1}(u)=\L_n(u^2)-(u-\ep)^2,\quad F_{n,2}(t,s,r)=-\Bigl(\L_n(s^2-t^2)+(t-i\ep)^2+r^2\Bigr),\\
&F_n(\mathbf{x})=F_{n,1}(u)+F_{n,2}(t,s,r)-\ep v^2-2\ep w^2,\quad 
\Phi_n(\mathbf{x})=\dfrac{r-it-\ep}{\sqrt{v^2+4r-4it-4\ep}}\cdot\dfrac{u}{\sqrt{2u-w^2}}\cdot g_n(u^2,s^2-t^2).
\end{split}
\end{equation*}

\noindent
Next we analyse $F_{n}(\mathbf{x})$. Observe that  
$
\lim\limits_{u\to +\infty} F_{n,1}(u)=-\infty
$
and $
\partial_{u} F_{n,1}(u)=2u\cdot\tr G(u^2)-2u+2\ep
$,
which means that $u=x_{\ep,n}$ is the maximum point of $F_{n,1}(u)$, where $x_{\ep,n}$ is the positive solution of \eqref{eq_x_en} defined in Subsection~\ref{sbsc:saddle}.

We can expand $F_{n,1}(u)$ for $u$ lying in some neighbourhood of $x_{\ep,n}$:
\begin{equation}\label{F_1_expand}
F_{n,1}(u)=\L_n(x_{\ep,n}^2)-(x_{\ep,n}-\ep)^2-\kappa_1(u-x_{\ep,n})^2+O\bigl((u-x_{\ep,n})^3\bigr),
\end{equation}
where 
\begin{equation}\label{def_k_1}
\kappa_1=\dfrac{\ep}{x_{\ep,n}}+\tr G^2(x_{\ep,n}^2)\cdot 2x_{\ep,n}^2.
\end{equation}
Bounds~\eqref{bound_x_en} imply that $\kappa_1\ge c>0$ uniformly in $n$, thus 
\begin{equation}\label{F_1_ineq_rstr}
\begin{split}
&F_{n,1}(u)\le F_{n,1}(x_{\ep,n})-c n^{-1}\log^2 n   \text{\quad when }|u-x_{\ep,n}|>n^{-1/2}\log n,
\end{split}
\end{equation}
for large $n$ and small $\ep$, where $c>0$ does not depend on $n$.
\bigskip

Dealing with $F_{n,2}(t,s,r)$, we start with the contour shift for $t$ and $r$.
Consider the function 
$$
h_n(t)=F_{n,2}(t,0,0)=-\L_n(-t^2)-(t-i\ep)^2.
$$ 
It is easy to see that $h_n(t)$ is analytic in the upper halfplane and 
 $t=ix_{\ep,n}$ is a stationary point of $h_n(t)$.
We can move the integration with respect to $t$ to a contour 
$$
L_t=L_-\cup L_0\cup L_+\subset\{z\colon \Im z\ge |\Re z|+\ep\}
$$ 
symmetric with respect to the imaginary axis, such that $L_0=[ix_{\ep,n}-\delta;ix_{\ep,n}+\delta]$, $\Re h_n(t)$ decreases on $[ix_{\ep,n},ix_{\ep,n}+\delta]$ and $\Re h_n(t)\le h_n(ix_{\ep,n})-\sigma$ for $t\in L_{\pm}$. One can check this using level lines of $\Re h_n(t)$ similarly to \cite[Lemma 4.1]{SS:sng}. Moreover, one can choose $\delta,\sigma>0$ independent of $n$ since for $\wt h_n(t)=\Re h_n(ix_{\ep,n}+t)$ we have $\wt h'_n(0)=0$, $\wt h''_n(0)=-2\kappa_1$, where $\kappa_1\ge c>0$ is defined in \eqref{def_k_1}, and $\wt h'''_n(t)$ is bounded uniformly in $n$ for small $t$. We should also ensure that $L_0$ lies inside $\{z\colon \Im z\ge |\Re z|+\ep\}$, which imposes the following condition: $\delta\le x_{\ep,n}-\ep$. Proposition~\ref{pro:sols_bound_int} and Proposition~\ref{pro:sols_bound_ext} yield that $x_{\ep,n}-\ep\ge c_\ep>0$, hence we can choose $\delta>0$ independent of $n$.

Also we deform the $r$-contour for each $t\in L_t$ as follows: 
$
\wt{\mathcal{R}}(t)=\mathcal{R}_1(t)\cup\mathcal{R}_2(t)$, where \linebreak
$\mathcal{R}_1(t)=[it+\ep,-\delta]$,
$\mathcal{R}_2(t)=\{r\colon r=-\delta+\rho, \rho\ge 0\}$.
Such a contour shift is allowed since for each fixed $t$ we have $|\Im r|\le C$ and thus $-\Re r^2\le C-|r|^2$ for big~$r$.

Next we prove that $(ix_{\ep,n},0,0)$ is a  maximum point of $\Re F_{n,2}(t,s,r)$ when $t\in L_t$, $s\in\mathbb{R}$,  $r\in\wt{\mathcal{R}}(t)$ that is `good enough' for the saddle point method. More precisely:
\begin{pro}
$(ix_{\ep,n},0,0)$ is the maximum point of $\Re F_{n,2}(t,s,r)$ when $s\in\mathbb{R}$, $t\in L_t$ and $r\in\wt{\mathcal{R}}(t)$. Moreover,
\begin{equation}\label{F_2_ineq_rstr}
\begin{split}
\Re F_{n,2}(t,s,r)\le  F_{n,2}(ix_{\ep,n},0,0)-cn^{-1}\log^2 n \text{\textnormal{\quad when}\quad} \max\{|t-ix_{\ep,n}|,|s|,|r|\}>n^{-1/2}\log n.
\end{split}
\end{equation}
\end{pro}

\begin{proof}
The statement above is a straightforward consequence of the following inequalities:
\begin{align}
\label{ineq_2}
&\Re F_{n,2}(t,s,r)\le \Re F_{n,2}(t,s,0) &  \quad\text{ when\ } &t\in L_t;\\
\label{ineq_3}
&\Re F_{n,2}(t,s,0)\le \Re F_{n,2}(t,0,0)&\quad\text{ when\ }&t\in L_t;\\
\label{ineq_4}
&\Re F_{n,2}(t,0,0)\le \Re F_{n,2}(ix_{\ep,n},0,0)-\sigma &\quad\text{ when\ }&t\in L_{\pm};\\
\label{ineq_5}
&\Re F_{n,2}(t,0,0)\le \Re F_{n,2}(ix_{\ep,n},0,0)&\quad\text{ when\ }&t\in L_{0};\\
\label{ineq_6}
&\Re F_{n,2}(t,s,0)\le \Re F_{n,2}(t,0,0)-c\epsilon^2&\quad\text{ when\ }&t\in L_{0},\ |t-ix_{\ep,n}|<\epsilon,\ |s|>\epsilon;\\
\label{ineq_7}
&\Re F_{n,2}(t,0,0)\le \Re F_{n,2}(ix_{\ep,n},0,0)-c\epsilon^2&\quad\text{ when\ }&t\in L_{0},\ |t-ix_{\ep,n}|>\epsilon;\\
\label{ineq_8}
&\Re F_{n,2}(t,s,r)\le \Re F_{n,2}(t,s,0)-c\epsilon^2 &\quad\text{ when\ }& t\in L_0,\ |r|>\epsilon.
\end{align}

Let us start from the proof of  \eqref{ineq_2}.
It suffices to prove that $\Re r^2\ge 0$. Set $t=t_1+it_2$, then $t_2-|t_1|\ge \ep>0$ for $t\in L_t$. For $r\in \mathcal{R}_2(t)$ the inequality is obvious. For $r\in\mathcal{R}_1(t)$ we have ${r=\alpha(-t_2+\ep+it_1)-(1-\alpha)\delta}$, thus
\begin{equation*}
\Re r^2=(-\alpha (t_2-\ep)-(1-\alpha)\delta)^2-(\alpha t_1)^2\ge \alpha^2((t_2-\ep)^2-t_1^2)\ge 0.
\end{equation*}

Next, we prove \eqref{ineq_3}.
Set $t=t_1+it_2$, then $t_2-|t_1|\ge \ep>0$ since $t\in L_t$. Also, $t^2=t_1^2-t_2^2+2it_1t_2$, thus
\begin{equation*}
\Re F_{n,2}(t,s,0)=-\dfrac{1}{2}\int \log((\lambda+t_2^2-t_1^2+s^2)^2+4t_1^2t_2^2)\,d\nu_{n,z}(\lambda)-\Re (t-i\ep)^2
\end{equation*}
Since $t_2^2-t_1^2>0$, then $\Re F_{n,2}(t,s,0)$ decreases for $s\in[0,+\infty)$, which gives us~\eqref{ineq_3}.

Notice that $\Re F_{n,2}(t,0,0)=\Re h_n(t)$. The inequalities $\Re h_n(t)\le \Re h_n(ix_{\ep,n})-\sigma$ for $t\in L_{\pm}$ and $\Re h_n(t)\le \Re h_n(ix_{\ep,n})$ for $t\in L_{0}$ imply that \eqref{ineq_4} and \eqref{ineq_5} hold.

The inequality \eqref{ineq_6} follows from the fact that $\Re F_{n,2}(t,s,0)$ decreases for $s\in[0,+\infty)$ and 
$$
\partial_{s}^2 \Re F_{n,2}(t,s,0)\bigr|_{s=0} =-2\Re\tr G(-t^2)<-1
$$
for $t$ lying in some neighbourhood of $ix_{\ep,n}$,
while \eqref{ineq_7} follows from the fact that $\Re F_{n,2}(t,0,0)=\Re h_n(t)$ decreases when $t\in [ix_{\ep,n},ix_{\ep,n}+\delta]$, and 
$\partial_{\tau}^2\Re h_n(ix_{\ep,n}+\tau)\bigr|_{\tau=0}=-2\kappa_1<-c$. 

Finally, we prove \eqref{ineq_8}.
It suffices to show that $\Re r^2\ge c\epsilon^2$. For $r\in\mathcal{R}_2(t)$ it is obvious. In case $r\in\mathcal{R}_1(t)$ set $t=ix_{\ep,n}+\tau$, $|\tau|<\frac{1}{2}x_{\ep,n}$, then $r=\alpha(-x_{\ep,n}+i\tau+\ep)-(1-\alpha)\delta$ and for small $\tau>0$,
\begin{equation*}
\begin{split}
\Re r^2&=(-\alpha(x_{\ep,n}-\ep)-(1-\alpha)\delta)^2-(\alpha\tau)^2\ge \alpha^2((x_{\ep,n}-\ep)^2-\tau^2)+(1-\alpha)^2\delta^2\ge c>0.
\end{split}
\end{equation*}

\end{proof}

It is easy to check that for $(t,s,r)$ lying in some neighbourhood of $(ix_{\ep,n},0,0)$ we have
\begin{equation}\label{F_2_expand}
\begin{split}
F_{n,2}(t,s,r)
=-\L_n(x_{\ep,n}^2)+(x_{\ep,n}-\ep)^2-\kappa_1 (t-ix_{\ep,n})^2-\kappa_2s^2-r^2+O(|s|^3+|t-ix_{\ep,n}|^3),
\end{split}
\end{equation}
where $\kappa_1$ defined in \eqref{def_k_1} and
\begin{equation}\label{def_k_2}
\kappa_2=\tr G(x_{\ep,n}^2).
\end{equation}
Bounds~\eqref{bound_x_en} imply that $\kappa_2\ge c>0$ uniformly in $n$.

Observe that \eqref{F_1_ineq_rstr} and \eqref{F_2_ineq_rstr} give us
\begin{equation*}
\begin{split}
\Re F_n(\mathbf{x})\le -cn^{-1}\log^2 n&\text{\quad when\quad} \max\{|u-x_{\ep,n}|,|t-ix_{\ep,n}|,|s|,|r|,|v|,|w|\}>n^{-1/2}\log n,
\end{split}
\end{equation*}
which allows us to restrict the integration to the neighbourhood
\begin{equation*}
\begin{split}
&U_{n}=\{\mathbf{x}\in \wt{\mathcal{V}}\colon |u-x_{\ep,n}|,\, |t-ix_{\ep,n}|,\,|s|,\, |r|,\,|v|,\, |w| <n^{-1/2}\log n\}.
\end{split}
\end{equation*}
with an error term $O(e^{-c\log^2 n})$.
Making the changes of variables $u=x_{\ep,n}+n^{-1/2}\wt u$, $t=ix_{\ep,n}+n^{-1/2}\wt t$, $s=n^{-1/2}\wt s$,  $r=n^{-1/2}\wt r$, $v=n^{-1/2}\wt{v}$, $w=n^{-1/2}\wt{w}$, using the expansions \eqref{F_1_expand}, \eqref{F_2_expand} and expanding the integrand, we obtain
\begin{equation}\label{I_g_n_expr}
\begin{split}
\I_{g_n}
&=g_n(x_{\ep,n}^2, x_{\ep,n}^2)
\cdot\dfrac{x_{\ep,n}-\ep}{\sqrt{4x_{\ep,n}-4\ep}}
\cdot\dfrac{x_{\ep,n}}{\sqrt{2x_{\ep,n}}}\times\\
&\times\dfrac{8}{\pi^3}\int_{\wt{U}_n} 
e^{-\kappa_1 \wt{u}^2-\kappa_1\wt{t}^2-\kappa_2\wt{s}^2-\wt{r}^2-\ep \wt{v}^2-2\ep\wt{w}^2}\,
d\wt{u}\,d\wt{t}\,d\wt{s}\,d\wt{r}\,d\wt{v}\,d\wt{w}\,+ O(n^{-1/2}\log^k n)=\\
&=C(\ep,n,z)\cdot g_n(x_{\ep,n}^2,x_{\ep,n}^2)+O(n^{-1/3}),
\end{split}
\end{equation}
where
$
C(\ep,n,z)=\sqrt{\dfrac{(x_{\ep,n}-\ep)x_{\ep,n}}{\kappa_1^2\kappa_2\ep^2}}.
$
Notice that $C(\ep,n,z)$ does not depend on $g_n$ and  $C(\ep,n,z)=O(1)$ for fixed $\ep>0$ as $n\to\infty$, since $\kappa_{1,2}\ge c>0$. Substituting $g_n(u_1^2+u_2^2,s^2-t^2)=\varphi(u_1^2+u_2^2,s^2-t^2,z,z)$ in \eqref{I_g_n_expr}, we obtain
\begin{equation*}
\I_{\varphi(z,z)}=C(\ep,n,z)\cdot \varphi(x_{\ep,n}^2,x_{\ep,n}^2,z,z)+O(n^{-1/3})
\end{equation*}
On the other hand, according to \eqref{Z_repr_aver}, 
\begin{equation*}
\I_{\varphi(z,z)}=\mathcal{I}\bigl( \varphi(u_1^2+u_2^2,s^2-t^2,z,z) e^{n\F_n(z,z)}\bigr)=\mathcal{Z}(\ep,\ep,z,z)=1.
\end{equation*}
Therefore, 
\begin{equation*}
\I_{g_n}=\dfrac{\I_{g_n}}{\I_{\varphi(z,z)}}=\dfrac{g_n(x_{\ep,n}^2,x_{\ep,n}^2)}{ \varphi(x_{\ep,n}^2,x_{\ep,n}^2,z,z)}+O(n^{-1/3}),
\end{equation*}
which gives us \eqref{form_<ge^nF>}.

\begin{rema}
In fact, one can write a better error term $O(n^{-1})$ as in usual Gaussian integral. This is due to the fact that integrals of the form $\int x^k e^{-ax^2}\,dx$ are equal to zero for odd $k$ and bounded for even $k$. However, at this point we are not interested in the best possible bound. We perform more precise analysis of the error term when $\ep$ depends on $n$ in Section~\ref{sc:rate}.

\end{rema}

\subsection{Asymptotic behaviour of $\I_k$ and $\rho_{\ep,n}(z)$}

Using the formula \eqref{form_<ge^nF>}, we can now easily obtain the following result.
\begin{pro}\label{pro:for_I_k}
For $\I_2,\I_3,\I_4$ defined in \eqref{def_I_k} and $z\notin \partial D$ we have
\begin{equation}\label{for_I_k}
\begin{split}
&\I_2=x^2_{\ep,n}\cdot\tr G(x_{\ep,n}^2)\wt{G}(x_{\ep,n}^2)+O(n^{-\alpha}),\quad
\I_3=O(n^{-\alpha}),\quad \I_4=O(n^{-1}).
\end{split}
\end{equation}
\end{pro}

\begin{proof}
One can recall the definition of $\I_k$ and apply \eqref{form_<ge^nF>} for the following functions:
\begin{equation*}
\begin{split}
g_{n,2}(x,y)&=\bigl((\partial_{\bar{z}}+\partial_{\bar{z}_1}) \partial_{z_1}\F_n(z,z_1)\cdot \varphi(z,z_1)\bigr)\Bigr|_{z_1=z}=
x\cdot\tr G(x)\wt{G}(x)\cdot\varphi(x,y,z,z);\\
g_{n,3}(x,y)&=\bigl(\partial_{z_1}\varphi(z,z_1)\cdot(\partial_{\bar{z}}+\partial_{\bar{z}_1})\F_n(z,z_1)\bigr)\Bigr|_{z_1=z}=\\
&=\Bigl(\tr\,(z-A_n)G(x)-\tr\,(z-A_n)G(y)\Bigr)\partial_{z_1}\varphi(x,y,z,z);\\
g_{n,4}(x,y)&=\dfrac{1}{n}\bigr((\partial_{\bar{z}}+\partial_{\bar{z}_1})\partial_{z_1}\varphi(z,z_1)\bigr)\Bigr|_{z_1=z}=O(n^{-1}).
\end{split}
\end{equation*}

\end{proof}

Next, we are going to find an asymptotic formula for $\I_1$. If we set
\begin{equation*}
g_{n,1}=\partial_{z_1}\F_n(z,z)\Bigl(\partial_{\bar{z}}+\partial_{\bar{z}_1}\Bigr)\Bigl( \varphi(z,z_1)\,e^{n\F_n(z,z_1)}\Bigr)\cdot e^{-n\F_n(z,z)}
\end{equation*}
then we cannot apply \eqref{form_<ge^nF>} for $g_n=g_{n,1}$ since $g_{n,1}$ is not bounded as $n\to\infty$. However, we will implement a method similar to the one in Subsection~\ref{sbsc:<ge^nF>}.

\begin{pro}\label{pro:for_I_1}
For $\I_1$ defined in \eqref{def_I_k} and $z\notin \partial D$ we have
\begin{equation}\label{for_I_1}
\I_1=\dfrac{\Bigl|\tr(A_n-z)G^2(x_{\ep,n}^2)\Bigr|^2}{\tr G^2(x_{\ep,n}^2)+\ep/(2x_{\ep,n}^3)}+O(n^{-\alpha}).
\end{equation}
\end{pro}

\begin{proof}
The idea of the proof is to show that the integrand of $\I_1$ has a zero of the second order at the saddle point, which will neutralize an extra multiplier $n$ before the integral.

Denote $p(x)=\tr\bigl((z-A_n)G(x)\bigr)$. Then $\ol{p(x)}=\tr\bigl((\bar{z}-A_n^*)G(x)\bigr)$, $\partial_{z_1} \F_n(z,z)=\ol{p(u_1^2+u_2^2)}$ and
\begin{equation}\label{I_1_two_summands}
\begin{split}
\I_1
&=\mathcal{I}\Bigl( \ol{p(u_1^2+u_2^2)}\,\Bigl(\partial_{\bar{z}}+\partial_{\bar{z}_1}\Bigr)\Bigl( \varphi(z,z_1)\,e^{n\F_n(z,z_1)}\Bigr) \Bigr)\Bigr|_{z_1=z}
\end{split}
\end{equation}
We start with creating an extra root of the integrand at the saddle point. 
The identity~\eqref{<phi_exp_nF>=1} implies that
$
\left(\partial_{\bar{z}}+\partial_{\bar{z}_1}\right)\mathcal{Z}(\ep,\ep,z,z_1)\Bigr|_{z=z_1}=0,
$
which can be rewritten as
\begin{equation}\label{zero_equal}
\mathcal{I}\Bigl( \Bigl(\partial_{\bar{z}}+\partial_{\bar{z}_1}\Bigr)\Bigl( \varphi(z,z_1)\,e^{n\F_n(z,z_1)}\Bigr) \Bigr)\Bigr|_{z_1=z}=0.
\end{equation}
Subtracting \eqref{zero_equal} multiplied by $\ol{p(x_{\ep,n}^2)}$ from \eqref{I_1_two_summands}, we get
\begin{equation*}
\I_1=\mathcal{I}\Bigl( \bigl(\,\ol{p(u_1^2+u_2^2)}-\ol{p(x_{\ep,n}^2)}\,\bigr)\Bigl(\partial_{\bar{z}}+\partial_{\bar{z}_1}\Bigr)\Bigl( \varphi(z,z_1)\,e^{n\F_n(z,z_1)}\Bigr) \Bigr)\Bigr|_{z_1=z}
\end{equation*} 
Observe that
\begin{equation*}
\Bigl(\partial_{\bar{z}}+\partial_{\bar{z}_1}\Bigr)\Bigl( \varphi(z,z_1)\,e^{n\F_n(z,z_1)}\Bigr)\Bigr|_{z_1=z}
=\Bigl(n\left(\partial_{\bar z} \F_n+\partial_{\bar{z}_1} \F_n\right)\varphi +\partial_{\bar z} \varphi+\partial_{\bar{z}_1} \varphi\Bigr)\Bigr|_{z_1=z}e^{n\F_n(z,z)},
\end{equation*}
where 
\begin{equation}\label{Tr-Tr}
\left(\partial_{\bar{z}} \F_n+\partial_{\bar{z}_1} \F_n\right)\Bigr|_{z_1=z}
=p(u_1^2+u_2^2)-p(s^2-t^2).
\end{equation}
We can rewrite
\begin{equation}\label{I_1_int_form}
\I_1=n\mathcal{I}\Bigl(\bigl(\,\ol{p(u_1^2+u_2^2)}-\ol{p(x_{\ep,n}^2)}\,\bigr)\,\bigl(p(u_1^2+u_2^2)-p(s^2-t^2)+O(n^{-1})\bigr) \varphi(z,z)e^{n\F_n(z,z)}\Bigr)
\end{equation}
We can move contours, restrict the integration and change the variables as in Subsection~\ref{sbsc:<ge^nF>}. After the change of variables we get
\begin{equation*}
\begin{split}
\I_1=&n\int d\wt{u}\,d\wt{t}\,d\wt{s}\,d\wt{r}\,d\wt{v}\,d\wt{w}\, (\,\ol{p_1(u^2)}-\ol{p_1(x_{\ep,n}^2)}\,\bigr)\times\\
&\times\Bigl( \bigl(p(u^2)-p(x_{\ep,n}^2)\bigr)-\bigl(p(s^2-t^2)-p(x_{\ep,n}^2)\bigr)\Bigr) \Phi_n(u,t,s,r,v,w)\times\\
&\times\exp\{-\kappa_1\wt{u}^2-\kappa_1\wt t^2-\kappa_2 \wt s^2-\wt r^2-\ep\wt v^2-2\ep\wt{w}^2+O(n^{-1/2}\log^k n)\}
\end{split}
\end{equation*}
for some $\Phi_n$. Observe that $\ol{p(u^2)}-\ol{p(x_{\ep,n}^2)}$, $p(u^2)-p(x_{\ep,n}^2)$ and $p(s^2-t^2)-p(x_{\ep,n}^2)$ have zeros of the first order at the saddle point, thus the multiplier $n$ before the integral vanishes. Taking into account that
\begin{equation*}\begin{split}
&p(u^2)-p(x_{\ep,n}^2)=n^{-1/2}\gamma\wt u+O(n^{-1}\log^k n);\\
&p(s^2-t^2)-p(x_{\ep,n}^2)=-n^{-1/2} i\gamma\wt t+O(n^{-1}\log^k n),
\end{split}
\end{equation*}
where $\gamma=-2x_{\ep,n}\cdot\tr (z-A_n)G^2(x_{\ep,n}^2)$, we obtain
\begin{equation}\label{I_1_fin}
\begin{split}
\I_1=& \Phi_n(x_{\ep,n},ix_{\ep,n},0,0,0,0)\cdot\int d\wt{u}\,d\wt{t}\,d\wt{s}\,d\wt{r}\,d\wt{v}\,d\wt{w}\,
(\,\ol{\gamma}\gamma\cdot\wt{u}^2+i\ol{\gamma}\gamma\cdot\wt u\wt t)\times\\
&\times\exp\{-\kappa_1\wt{u}^2-\kappa_1\wt t^2-\kappa_2 \wt s^2-\wt r^2-\ep\wt v^2-2\ep\wt{w}^2\}+O(n^{-1/2}\log^k n).
\end{split}
\end{equation}
Since the Gaussian integral $\int x^ne^{-kx^2}\, dx$ equals zero for odd $n$, we can omit the summand $i\ol{\gamma}\gamma\cdot\wt u\wt t$. Now recall that \eqref{<phi_exp_nF>=1} holds and we can write
\begin{equation}\label{1_expand}
\begin{split}
1=&\mathcal{I}\bigl( \varphi(z,z) e^{n\F_n(z,z)}\bigr)=\Phi_n(x_{\ep,n},ix_{\ep,n},0,0,0,0)\times\\
&\times\int d\wt{u}\,d\wt{t}\,d\wt{s}\,d\wt{r}\,d\wt{v}\,d\wt{w}\,\exp\{-\kappa_1\wt{u}^2-\kappa_1\wt t^2-\kappa_2 \wt s^2-\wt r^2-\ep\wt v^2-2\ep\wt{w}^2\}+O(n^{-1/2}\log^k n).
\end{split}
\end{equation}
Dividing \eqref{I_1_fin} by \eqref{1_expand}, we obtain
\begin{equation*}
\begin{split}
\I_1&=\dfrac{\ol{\gamma}\gamma\int \wt u^2 \exp\{-\kappa_1\wt u^2\}}{\int \exp\{-\kappa_1\wt u^2\}}+O(n^{-1/3})=\dfrac{|\gamma|^2}{2\kappa_1}+O(n^{-1/3})=\\
&=\dfrac{\bigl|\tr(A_n-z)G^2(x_{\ep,n}^2)\bigr|^2}{\tr G^2(x_{\ep,n}^2)+\ep/(2x_{\ep,n}^3)}+O(n^{-1/3}).
\end{split}
\end{equation*}

\end{proof}

\noindent
Substituting \eqref{for_I_k} and \eqref{for_I_1} into the formula from Proposition~\ref{pro:int_rep_dens}, we obtain
\begin{equation}\label{asympt_rho_ep,n}
\pi\cdot\bar{\rho}_{\ep,n}(z)=\dfrac{\Bigl|\tr(A_n-z)G^2(x_{\ep,n}^2)\Bigr|^2}{\tr G^2(x_{\ep,n}^2)+\ep/(2x_{\ep,n}^3)}+x^2_{\ep,n}\cdot\tr G(x_{\ep,n}^2)\wt{G}(x_{\ep,n}^2)+\dfrac{\beta(\ep,n)}{n^{\alpha}},
\end{equation}
where $|\beta(\ep,n)|\le C(\ep)$ as $n\to\infty$.

\section{Formula for \boldmath{$\rho(z)$}}

We continue to implement the plan described in Remark~\ref{rm:to_prove}. In this section, we derive a formula for $\rho(z)$ using asymptotic formula~\eqref{asympt_rho_ep,n} for $\bar{\rho}_{\ep,n}(z)$.

\subsection{Case $z\in\Int D$}

We want to take the limit in \eqref{asympt_rho_ep,n} as $n\to\infty$ for fixed $\ep>0$ and $z\in\Int D$. Recall that ${\lim\limits_{n\to\infty} x_{\ep,n}=x_\ep}$, $\lim\limits_{n\to\infty} \tr(A_n-z)G^2(x)=T_1(z,x)$, $\lim\limits_{n\to\infty} \tr G(x)\wt{G}(x)=T_2(z,x)$, $\lim\limits_{n\to\infty}\tr G^2(x)=\int (\lambda+x)^{-2}\,d\nu_z(\lambda)$, $x>0$, where $T_1,T_2$ and $\nu_z$ are defined in Subsection~\eqref{sbsc:bas_notat}. From the fact that all the functions in those limits are analytic for $x>0$ one can easily obtain that   
\begin{equation*}
\begin{split}
&\lim\limits_{n\to\infty} \tr(A_n-z)G^2(x_{\ep,n}^2)=T_1(z,x_{\ep}^2),\\
&\lim\limits_{n\to\infty} \tr G(x_{\ep,n}^2)\wt{G}(x_{\ep,n}^2)=T_2(z,x_\ep^2),\\
&\lim\limits_{n\to\infty}\tr G^2(x_{\ep,n}^2)=\int (\lambda+x_\ep^2)^{-2}\,d\nu_z(\lambda).
\end{split}
\end{equation*}
Thus, for $\rho_{\ep}(z)=\lim\limits_{n\to\infty} \bar{\rho}_{\ep,n}(z)$ we obtain the following identity:
\begin{equation}\label{rho_ep_form}
\pi\cdot\rho_{\ep}(z)=\dfrac{|T_1(z,x^2_\ep)|^2}{\int (\lambda+x_\ep^2)^{-2}\,d\nu_z(\lambda)+\ep/(2x_{\ep}^3)}+x^2_{\ep}\cdot T_2(z,x^2_\ep).
\end{equation}
We are left to find the limit $\rho(z)=\lim\limits_{\ep\to 0} \rho_{\ep}(z)$. 
Recall that $\lim\limits_{\ep\to 0}x_{\ep}=x_0$. According to Subsection~\ref{sbsc:bas_notat}, the functions $T_{1,2}(z,x)$ are analytic in $z,x$ for $x>0$, hence they are continuous for $x>0$. Obviously, the function ${\int (\lambda+x^2)^{-2}\,d\nu_z(\lambda)}$ is also continuous for $x>0$. Hence,
\begin{equation}\label{rho_form}
\pi\cdot\rho(z)=\dfrac{|T_1(z,x^2_0)|^2}{\int (\lambda+x_0^2)^{-2}\,d\nu_z(\lambda)}+x^2_{0}\cdot T_2(z,x^2_0),
\end{equation}
which shows that $\rho(z)$ is equal to $\rho_\mu(z)$ defined in \eqref{rho_main_form} for $z\in\Int D$.

Also we need to check that $\rho_{\ep}(z)$ is bounded uniformly in $\ep\le\ep_0$ and $z\in\Int D$. To estimate $T_j(z,x^2_{\ep,n})$, it is enough to obtain upper bounds on $\bigl|\tr (A_n-z)G^2(x_{\ep,n}^2)\bigr|$ and $\tr G(x_{\ep,n}^2)\wt{G}(x_{\ep,n}^2)$. Observe that $\dist(\sigma_\epsilon,\partial D)>0$, where $\sigma_\epsilon$ is defined in (C3). Now one can easily obtain upper bounds on ${\bigl|\tr (A_n-z)G^2(x_{\ep,n}^2)\bigr|}$, $\tr G(x_{\ep,n}^2)\wt{G}(x_{\ep,n}^2)$ and a lower bound on $\int (\lambda+x_\ep^2)^{-2}\,d\nu_z(\lambda)$, using (C2), Proposition~\ref{pro:sols_bound_int} for $z\in\sigma_\epsilon$ (taking $E_{in}=\sigma_\epsilon$) and (C3) for $z\notin \sigma_{\epsilon}$.

\subsection{Case $z\in \mathbb{C}\setminus\ol{D}$}

Similarly to the case $z\in\Int D$ we obtain 
\begin{equation*}
\pi\cdot\rho_{\ep}(z)=\dfrac{|T_1(z,x^2_\ep)|^2}{\int (\lambda+x_\ep^2)^{-2}\,d\nu_z(\lambda)+\ep/(2x_{\ep}^3)}+x^2_{\ep}\cdot T_2(z,x^2_\ep).
\end{equation*}
To find the limit $\rho(z)=\lim\limits_{\ep\to 0} \rho_{\ep}(z)$, use the bound
\begin{equation}\label{rho_ep_out}
\pi\cdot\rho_{\ep}(z)\le\dfrac{|T_1(z,x^2_\ep)|^2}{\ep/(2x_{\ep}^3)}+x^2_{\ep}\cdot T_2(z,x^2_\ep)
=|T_1(z,x^2_\ep)|^2\cdot 2x_{\ep}^2\cdot\dfrac{x_\ep}{\ep}+x^2_{\ep}\cdot T_2(z,x^2_\ep).
\end{equation}
According to Proposition~\ref{pro:sols_bound_ext}, $x_\ep\le (1+K_0)\ep$ for a fixed $z\in\mathbb{C}\setminus\ol{D}$. 
Using condition (C3) one can show that $|T_{1,2}(z,x_\ep^2)|$ are bounded uniformly in $\ep$.  Then \eqref{rho_ep_out} gives us
\begin{equation}\label{rho_ep_estim}
|\rho_{\ep}(z)|\le C\ep^2
\end{equation} 
for some $C>0$, which shows that 
$
\rho(z)=\lim\limits_{\ep\to 0} \rho_{\ep}(z)=0,
$
and $\rho(z)$ is equal to $\rho_\mu(z)$ defined in \eqref{rho_main_form} for $z\in\mathbb{C}\setminus D$.

Again, we need to prove that $\rho_{\ep}(z)$ is bounded uniformly in $\ep\le\ep_0$ and $z\in \mathbb{C}\setminus \overline{D}\colon |z|\le C$. Proposition~\ref{pro:sols_bound_ext} and Proposition~\ref{pro:sols_bound_boundary} imply that $x_{\ep}\le C\ep^{1/3}$ for $z\notin D$, $|z|\le C$. We obtain a bound
\begin{equation}
\pi\cdot\rho_{\ep}(z)\le\dfrac{|T_1(z,x^2_\ep)|^2}{\ep/(2x_{\ep}^3)}+x^2_{\ep}\cdot T_2(z,x^2_\ep)
\le C|T_1(z,x^2_\ep)|^2+C\ep^{2/3} T_2(z,x^2_\ep).
\end{equation}
The fact that $|T_j(z,x^2_{\ep})|$ are uniformly bounded follows from (C3) and (C2).

\section{Proof of Proposition \ref{l:uni_conv}}\label{sc:uni_conv_prf}

Recall that we have a random matrix $X_n=A_n+H_n$ with eigenvalues $z_1,\ldots,z_n$ and 
$$
Y(z)=(X_n-z)(X_n-z)^*.
$$ 
Also, we have a function $h(z)\in C_c^2(\mathbb{C})$ with a compact support $E$. According to \eqref{potent_formula}, it suffices to prove that
\begin{equation*}
\lim_{\ep\to 0} \int_E  h(z)\cdot  \dfrac{1}{4\pi n}\Delta\E_{H_n}\Bigl\{\log\det (Y(z)+\ep^2)\Bigr\}\,\dz = \int_E h(z)\cdot \dfrac{1}{4\pi n}\Delta  \E_{H_n}\Bigl\{\log\det Y(z)\Bigr\}\,\dz
\end{equation*}
unifromly in $n$ for $n\ge n_0$, which is equivalent to the fact that
\begin{equation}
\label{l:lim_val_toprove}
\lim_{\ep\to 0}\dfrac{1}{4\pi n}\E_{H_n}\Bigl\{ \int_E \Delta h(z)\cdot  \log\det (Y(z)+\ep^2)\,\dz \Bigr\}= \dfrac{1}{4\pi n}\E_{H_n}\Bigl\{\int_E \Delta h(z)\cdot \log\det Y(z)\,\dz\Bigr\}
\end{equation}
uniformly in $n$ for $n\ge n_0$.

Let us split the proof into two parts: in Subsection~\ref{sbsc:lim_val} we check that the convergence~\eqref{l:lim_val_toprove} holds for each $n$, and in Subsection~\ref{sbsc:uni_conv} we check that the convergence is uniform in $n$.

\subsection{Pointwise convergence}\label{sbsc:lim_val}

Since $\log\det (Y(z)+\ep^2)\to \log\det Y(z)$ pointwise as $\ep\to 0$, it suffices to estimate the integrand in~\eqref{l:lim_val_toprove} by some integrable over $E$ function such that the integral has a finite $H_n$-expectation.

Since $h$ is smooth with compact support, we have $|\Delta h(z)|\le C$, and it is sufficient to estimate $\log\det (Y(z)+\ep^2)$. For $\ep\in(0;1)$ we have 
\begin{equation}\label{logdet_est1}
|\log\det (Y(z)+\ep^2)|\le |\log\det Y(z)|+\log\det (Y(z)+1).
\end{equation}
Obviously,  $\log\det (Y(z)+1)\le C$ for a fixed $X_n$ and $z\in E$, while
$
|\log\det Y(z)|\le \sum\limits_{j=1}^n \bigl|\log |z-z_j|^2\bigr|,
$
which means that $|\log\det Y(z)|$
has integrable singularities at $z_1,\ldots,z_n$. This shows that for a fixed $X_n$, $\log\det (Y(z)+\ep^2)$ is dominated by some integrable over $E$ function, and thus
\begin{equation*}
\lim_{\ep\to 0}\int_E \Delta h(z)\cdot  \log\det (Y(z)+\ep^2)\,\dz = \int_E \Delta h(z)\cdot \log\det Y(z)\,\dz.
\end{equation*}
Now we need to estimate $\int_E \Delta h(z)\cdot  \log\det (Y(z)+\ep^2)\,\dz$. We can write $\log\det Y(z)=A+B$, where 
$$
A=\sum\limits_{j\colon |z-z_j|>1} \log |z-z_j|^2,\quad 
B=\sum\limits_{j\colon |z-z_j|\le 1} \log |z-z_j|^2.
$$ 
Then $A>0$, $B\le 0$ and $|A+B|\le A-B=A+B-2B$, which means that
\begin{equation*}
|\log \det Y(z)|\le \log\det Y(z)-2\sum_{j\colon |z-z_j|\le1} \log |z-z_j|^2.
\end{equation*}
Using the inequality above, \eqref{logdet_est1} and an obvious inequality $\log x<x$, after integrating we obtain
\begin{equation*}
\begin{split}
\Big|\int_E \Delta h(z)\cdot  \log\det (Y(z)+\ep^2)\,\dz\Big|\le 
C\Bigl(&\int_E \det Y(z)\,\dz+\int_E \det (Y(z)+1)\,\dz-\\
&-2\sum_{j=1}^n\int_{|z-z_j|\le 1} \log|z-z_j|^2\,\dz\Bigr)\le\\
\le C\Bigl(&\int_Q \det Y(z)\,\dz+\int_Q \det (Y(z)+1)\,\dz+2n\pi\Bigr),
\end{split}
\end{equation*}
where $Q=\{z\in\mathbb{C}: |\Re z|\le  r,\ |\Im z|\le r\}$ is a square containing $E$. Is is easy to see that $\int_Q \det Y(z)\,\dz$ and $\int_Q \det (Y(z)+1)\,\dz$ are polynomials depending on the entries of $X_n$, i.e. on $\{a_{ij}\}$ and $\{h_{ij}\}$. Since $h_{ij}$ are independent Gaussian random variables, we have $\E_{H_n}\{\prod |h_{ij}|^{k_{ij}}\}<\infty$, thus 
\begin{equation*}
\E_{H_n}\Bigl\{\int_Q \det Y(z)\,\dz+\int_Q \det (Y(z)+1)\,\dz+2n\pi\Bigr\}<\infty.
\end{equation*}
Dominated convergence theorem then gives us \eqref{l:lim_val_toprove} for fixed $n$, which finishes the proof.

\subsection{Uniform convergence}\label{sbsc:uni_conv}

It suffices to prove that $\mathsf{\Phi}(\ep,n,z)=\dfrac{1}{n}\E_{H_n}\Bigl\{\log\det (Y(z)+\ep^2)\Bigr\}$ converges uniformly for $n\ge n_0$ and $z\in E$ as $\ep\to 0$.
Set
$$
T(\ep,n,z)=\partial_{\ep}\mathsf{\Phi}(\ep,n,z)=\E_{H_n}\Bigl\{2\ep\cdot\tr(Y(z)+\ep^2)^{-1}\Bigr\}
.$$
Similarly to \eqref{rho_en_form} one can prove that
$
T(\ep,n,z)=\dfrac{1}{n}\Bigl(\partial_{\ep_1}\mathcal{Z}(\ep,\ep_1,z,z)\Bigr)\Bigr|_{\ep_1=\ep},
$
where $\mathcal{Z}(\ep,\ep_1,z,z)$ is defined in \eqref{def_Z}.
By differentiating \eqref{I_rep_Z} with respect to $\ep_1$, we obtain the following integral representation:
\begin{equation}\label{T_epn_repr1}
\begin{split}
T(\ep,n,z)
=&-\dfrac{4n^3}{\pi^3}\int_0^\infty dR\int_{-\infty}^\infty dv\,du_1\,du_2\,ds\int_L dt\cdot \dfrac{R}{\sqrt{v^2+4R}}\cdot \varphi(u_1^2+u_2^2,s^2-t^2,z,z)\times\\
&\times (u_1+\ep)\cdot \exp\left\{n\Bigl(\L_n(u_1^2+u_2^2)-(u_1+\ep)^2-u_2^2\Bigr)\right\}\times\\
&\times  \exp\left\{-n\Bigl(\L_n(s^2-t^2)+(t-i\ep)^2+(R+it+\ep)^2+\ep v^2\Bigr)\right\},
\end{split}
\end{equation}
Suppose that $T(\ep,n,z)\le \dfrac{C}{\ep^{1-\alpha}}$ for $\ep\le \ep_0$, $n\ge n_0$, $z\in E$ and some fixed $\alpha,C>0$. Then, for $\ep_1<\ep_2\le \ep_0$ we have
\begin{equation*}
|\mathsf{\Phi}(\ep_2,n,z)-\mathsf{\Phi}(\ep_1,n,z)|=\Bigl|\int_{\ep_1}^{\ep_2} T(\ep,n,z) \,d\ep\Bigr|\le C \int_{\ep_1}^{\ep_2} \dfrac{d\ep}{\ep^{1-\alpha}}\le \dfrac{C}{\alpha}\,\ep_2^\alpha,
\end{equation*}
which means that $\mathsf{\Phi}(\ep,n,z)$ converges uniformly for $n\ge n_0$, $z\in E$ as $\ep\to 0$.

As we see, it suffices to prove the following facts:

\begin{thm}\label{thm:uni_bound_far}
Fix an arbitrary $d>0$ and set $E_d:=\{z\in E\mid \dist(z,\partial D)\ge d\}$. Then there exist $n_0$, $\ep_0>0$ and $C>0$ such that 
$$
\ep^{1/2}\,|T(\ep,n,z)|\le C
$$ 
for $n\ge n_0$, $0<\ep\le\ep_0$ and $z\in E_d$.
\end{thm}

\begin{thm}\label{thm:uni_bound_near}
There exist $d>0$, $n_0$, $\ep_0>0$ and $C>0$ such that 
$$
\ep^{5/6}\,|T(\ep,n,z)|\le C
$$
 for $n\ge n_0$, $0<\ep\le\ep_0$ and $z\in E\setminus E_d$.
\end{thm}

\noindent
\textbf{Outline of the proof.}
We split the proof into several cases depending on the size of $\ep$ with respect to $n$ and the location of $z$. In each of the parts the plan similar to the one in Subsection~\ref{sbsc:<ge^nF>}  is implemented. We make the following steps:
\begin{enumerate}

\item
We make minor changes of variables in \eqref{T_epn_repr1}. In particular, for small $\ep$ we take $r=R+it$. After that we shift the $t$-contour and $r$-contour so that $L_t\subset \{z\colon \Im z\ge |\Re z|\}$ and $\wt{\mathcal{R}}(t)=\mathcal{R}_1(t)\cup\mathcal{R}_2(t)$, where $\mathcal{R}_1(t)=[it;0]$ and $\mathcal{R}_2(t)=\{\rho: \rho>0\}$. Then we have $\Re r^2\ge 0$ and $\Re (r-it)\ge 0$. For large $\ep$ we take $r=R+it+\ep$, shift the $t$-contour and $r$-contour so that $L_t\subset \{z\colon \Im z\ge |\Re z|+\ep\}$ and $\wt{\mathcal{R}}(t)=\mathcal{R}_1(t)\cup\mathcal{R}_2(t)$, where $\mathcal{R}_1(t)=[it+\ep;0]$ and $\mathcal{R}_2(t)=\{\rho: \rho>0\}$. Then we have $\Re r^2\ge 0$ and $\Re (r-it-\ep)\ge 0$.

\item 
Then it suffices to prove that
$n^\alpha\ep^{\beta}\int_{\mathcal{V}} \Phi_n(\mathbf{x})\,e^{n F_n(\mathbf{x})}\,d\mathbf{x}\le C$ uniformly in $n,\ep,z$, where $\mathbf{x}=(u,t,s,r)$ or $(u_1,u_2,t,s,r)$, $\mathcal{V}$ is a product of certain contours and $\Phi_n$, $F_n$ are some functions.

\item
We prove that $n F_n(\mathbf{x})\le -c\log^2 n$ for $\mathbf{x}\in\mathcal{V}$ lying outside of the neighbourhood 
$$
U=\{|u-x_*|<n^{-\alpha_1}\log^{\beta_1} n,\ |s|<n^{-\alpha_2}\log^{\beta_2} n,\ |t-ix_*|<n^{-\alpha_3}\log^{\beta_3} n, r\in \mathcal{R}_{1+}(t)\}
$$ 
of the saddle point $u=x_*$,  $t=ix_*$, $s=0$, $r=0$, where $\mathcal{R}_{1+}(t)=\mathcal{R}_1(t)\cup [0;n^{-1/2}\log n]$. This allows us to restrict the integration to the neighbourhood $U$, if $n^\alpha\ep^\beta\le Cn^\gamma$ for some $\gamma$.

\item 
We make a change $u=x_*+n^{-\alpha_1}\wt{u}$, $t=ix_*+n^{-\alpha_3}\wt{t}$, $s=n^{-\alpha_2}\wt{s}$  and estimate $\Phi_n(\mathbf{x})$ by expanding it into Taylor series:
$$
\Phi_n(\mathbf{x})\le f(\ep,x_*,n)\cdot \mathcal{P}(\wt{u},\wt{t},\wt{s}),
$$
where $\mathcal{P}(\wt{u},\wt{t},\wt{s})$ stands for an arbitrary polynomial in $|\wt{u}|$, $|\wt{s}|$, $|\wt{t}|$. We also expand $nF_n(\mathbf{x})$:
$$
nF_n(\mathbf{x})\le -a_1\wt{u}^{k_1}-a_2\wt{t}^{k_2}-a_3\wt{s}^2-nr^2+O(n^{-1/4}\log^k n),
$$
where $k_j\in\{2;4\}$ and $a_j$ are bounded from below uniformly by some positive constant. This gives the bound
$$
n^\alpha\ep^{\beta}\int_{U} \Phi_n(\mathbf{x})\,e^{n F_n(\mathbf{x})}\,d\mathbf{x}
\le \dfrac{n^{\alpha}\ep^{\beta}}{n^{\alpha_1}n^{\alpha_2}n^{\alpha_3}}\,f(\ep,x_*,n)\cdot\int_{\wt{U}} \mathcal{P}(\wt{u},\wt{t},\wt{s})\, e^{-a_1\wt{u}^{k_1}-a_2\wt{t}^{k_2}-a_3\wt{s}^2-nr^2}\,d\wt{\mathbf{x}}
$$

\item 
Finally, we either estimate $\int_{\mathcal{R}_{1+}(t)} e^{-nr^2}\,dr$ as $C(|t|+n^{-1/2})$ simply by considering the length of $\mathcal{R}_1(t)$  or write
$$
\left|\int_{\mathcal{R}_{1}(t)} e^{-nr^2}\,dr\right|\le \int_0^{|t|} \exp\{-\Re(\tfrac{-t^2}{|t|^2})\,n\rho^2\}\,d\rho\le \dfrac{Cn^{-1/2}}{\sqrt{\Re(\tfrac{-t^2}{|t|^2})}}.
$$
The first bound is better if $x_*$ is small and the second one is better 
if $\mathcal{R}_1(t)$ is `far enough' from ${\{z\in\mathbb{C}\colon \arg z=\frac{3\pi}{4}\}}$. In both cases  we get a bound of the form 
$$
\int_{\mathcal{R}_{1+}(t)} e^{-nr^2}\,dr\le g(\ep,x_*,n)\mathcal{P}(\wt{u},\wt{t},\wt{s}),
$$ 
and this bound implies that
$$
n^\alpha\ep^{\beta}\int_{U} \Phi_n(\mathbf{x})\,e^{n F_n(\mathbf{x})}\,d\mathbf{x}
\le \dfrac{n^{\alpha}\ep^{\beta}}{n^{\alpha_1}n^{\alpha_2}n^{\alpha_3}}\cdot f(\ep,x_*,n)\,g(\ep,x_*,n)\cdot\int_{\wt{U}} \mathcal{P}(\wt{u},\wt{t},\wt{s})\, e^{-a_1\wt{u}^{k_1}-a_2\wt{t}^{k_2}-a_3\wt{s}^2}\,d\wt{\mathbf{x}}.
$$
Then it suffices to prove that $\dfrac{n^{\alpha}\ep^{\beta}}{n^{\alpha_1}n^{\alpha_2}n^{\alpha_3}}\cdot f(\ep,x_*,n)\,g(\ep,x_*,n)$ is bounded uniformly in $\ep,n,z$.

\end{enumerate}

Set $\wt{\varphi}(u,t,s)=\varphi(u^2,s^2-t^2,z,z)$, where $\varphi(x,y,z,z)$ is defined in \eqref{notat_I_rep_Z}. In order to obtain a bound on $\Phi_n(\mathbf{x})$, we need to expand $\wt{\varphi}(u,t,s)$ in the neighbourhood of the saddle point $(x_*,ix_*,0)$, where $x_*$ is either $x_{\ep,n}$ or $x_{0,n}$. This bounds on the derivatives of $\varphi(x,y,z,z)$ are given in the following lemma.

\begin{lem}\label{phi_bounds}
Let $\varphi_1(x,y)=\varphi_1(x,y,z,z)$, $\varphi_2(x,y)=\varphi_2(x,y,z,z)$ be the functions defined as in \eqref{notat_I_rep_Z}, and let $x_{0,n}$, $x_{\ep,n}$ be the roots of \eqref{eq_x_0n} and \eqref{eq_x_en} respectively as in Subsection~\ref{sbsc:saddle}. Then $\varphi_{1,2}(x,y)$ together with all their derivatives are bounded at $(x_{0,n}^2,x_{0,n}^2)$ and $(x_{\ep,n}^2,x_{\ep,n}^2)$. Moreover,
\begin{equation*}
\begin{split}
&\varphi_1(x_{0,n}^2,x_{0,n}^2)=0,\quad 
\partial_x\varphi_1(x_{0,n}^2,x_{0,n}^2)=
\partial_y\varphi_1(x_{0,n}^2,x_{0,n}^2)=O(x_{0,n}^2);\\
&\varphi_1(x_{\ep,n}^2,x_{\ep,n}^2)=O\Bigl(\frac{\ep^2}{x_{\ep,n}^2}+\ep x_{\ep,n}\Bigr),\quad
\partial_x\varphi_1(x_{\ep,n}^2,x_{\ep,n}^2)=
\partial_y\varphi_1(x_{\ep,n}^2,x_{\ep,n}^2)=O\Bigl(\frac{\ep}{x_{\ep,n}}+x_{\ep,n}^2\Bigr).
\end{split}
\end{equation*} 
\end{lem}

\begin{proof}
The first half of the statement is obvious. Using Remark~\ref{rm:phi_z_1=z} one can get
\begin{equation*}
\begin{split}
\varphi_{1}(x,y)=&\Big(1-\tr G(y)+x\,\tr G(x)G(y)\Big)^2-xy\,(\tr G(x)G(y))^2;\\
\partial_x\varphi_{1}(x,y)&=2\Big(1-\tr G(y)+x\,\tr G(x)G(y)\Big)\Big(\tr G(x)G(y)-x\,\tr G^2(x)G(y)\Big)-\\
&-y\,(\tr G(z,x)G(z,y))^2+2xy\,\tr G(x)G(y)\,\tr G^2(x)G(y);\\
\partial_y\varphi_{1}(x,y)&=2\Big(1-\tr G(y)+x\,\tr G(x)G(y)\Big)\Big(\tr G^2(y)-x\,\tr G(x)G^2(y)\Big)-\\
&-x\,(\tr G(x)G(y))^2+2xy\,\tr G(x)G(y)\,\tr G(x)G^2(y).
\end{split}
\end{equation*}
Now it is easy to obtain more precise bounds on $\varphi_1$ and its first derivatives at $(x_{0,n}^2,x_{0,n}^2)$ and $(x_{\ep,n}^2,x_{\ep,n}^2)$ using the identities above.

\end{proof}

\begin{rema}\label{rem:phi_bounds}
Lemma~\ref{phi_bounds} implies that for the function $\wt{\varphi}(u,t,s)=\varphi(u^2,s^2-t^2,z,z)$ we have \linebreak ${\wt{\varphi}(x_{0,n},ix_{0,n},0)=O(n^{-1})}$, the first order derivatives of $\wt{\varphi}$ at $(x_{0,n},ix_{0,n},0)$ are $O(x_{0,n}^3)$, the second order derivatives of $\wt{\varphi}$ at $(x_{0,n},ix_{0,n},0)$ are $O(x_{0,n}^2)$, the third order derivatives of $\wt{\varphi}$ at $(x_{0,n},ix_{0,n},0)$ are $O(x_{0,n})$ and the higher order derivatives are bounded. Also, $\wt{\varphi}(x_{\ep,n},ix_{\ep,n},0)=O\Bigl(\dfrac{\ep^2}{x_{\ep,n}^2}+\ep x_{\ep,n}\Bigr)$,  the first order derivatives of $\wt{\varphi}$ at $(x_{\ep,n},ix_{\ep,n},0)$ are $O\bigl(\ep+x_{\ep,n}^3\bigr)$, the second order derivatives of $\wt{\varphi}$ at $(x_{\ep,n},ix_{\ep,n},0)$ are $O\Bigl(\dfrac{\ep}{x_{\ep,n}}+x_{\ep,n}^2\Bigr)$, the third order derivatives of $\wt{\varphi}$ at $(x_{\ep,n},ix_{\ep,n},0)$ are $O(x_{\ep,n})$ and the higher order derivatives are bounded.

\end{rema}

\subsubsection{Bounds on \boldmath{$T(\ep,n,z)$} for \boldmath{$\ep<n^{-1}$}}

We start with the following result:

\begin{pro}\label{pro:uni_bound_far_in_sm}
Fix an arbitrary $d>0$, and set $E_{d,in}:=\{z\in E\mid z\in D,\ \dist(z,\partial D)\ge d\}$. Then there exist $n_0\in\mathbb{N}$ and $C>0$ such that 
$$
\ep^{1/2}\,|T(\ep,n,z)|\le C
$$ 
for $n\ge n_0$, $0<\ep<n^{-1}$ and $z\in E_{d,in}$.
\end{pro}

\begin{proof}
Since $E_{d,in}$ is a compact subset of $\Int D$, we can use the results of Subsection~\ref{sbsc:saddle} for $E_{in}=E_{d,in}$.

Set $\hep:=n\ep$, then $0<\hep<1$. Make a change of variables $(u_1,u_2)\to (u,\theta)$ where $u_1=u\cos\theta$, $u_2=u\sin\theta$, $u\in[0,\infty)$, $\theta\in[0,2\pi]$. The Jacobian of this change $J=u$. Also we can make a change $r=R+it\in\mathcal{R}(t)$, where $\mathcal{R}(t)=\{-it+\tau,\ \tau>0\}$ and integrate with respect to $\theta$. Then we change the $t$-contour and $r$-contour to $L_t$ and $\wt{\mathcal{R}}(t)$ respectively, which are defined further in \eqref{R_cont_l:uc_int}, so that $\Re (r-it)\ge 0$ for $r\in\wt{\mathcal{R}}(t)$, and thus $\Bigl|\dfrac{r-it}{\sqrt{v^2+4(r-it)}}\Bigr|\le\dfrac{|\sqrt{r-it}|}{2}$. After integrating with respect to $v$ we obtain
\begin{equation}\label{T_repr_uc_1}
\begin{split}
&\ep^{1/2}\,|T(\ep,n,z)|\le Cn^{5/2}\int_{\mathcal{V}}\Phi_n(\mathbf{x})\,e^{n\Re F_n(\mathbf{x})}\,d\mathbf{x},\\
&\text{where}\quad\mathbf{x}=(u,t,s,r),\quad \mathcal{V}=[0,+\infty)\times \mathbb{R}\times L_t\times\wt{\mathcal{R}}(t),\\
&F_{n,1}(u)=\L_\n(u^2)-u^2,\quad F_{n,2}(t,s,r)=-\L_\n(s^2-t^2)-t^2-r^2,\\
&F_\n(\mathbf{x})=F_{n,1}(u)+F_{n,2}(t,s,r),\quad
\wt{\varphi}(u,t,s)=\varphi(u^2,s^2-t^2,z,z),\\
&\Phi_\n(u,t,s,r)=\sqrt{|r-it|}\cdot |\wt{\varphi}(u,t,s)|\cdot (u^2I_1(2 u)+\dfrac{1}{n}uI_0(2 u)).
\end{split}
\end{equation}
We can study $F_{n,1}$ and $F_{n,2}$ as in Subsection \ref{sbsc:<ge^nF>}.
Notice that $F_{n,1}'(u)=2u(\tr G(u^2)-1)$ has exactly one positive root $x_{0,\n}$, thus $u=x_{0,\n}$ is the maximum point of $F_{n,1}(u)$. 

For $F_{n,2}(t,s,r)$ consider $h_n(t)=F_{n,2}(t,0,0)=-\L_\n(-t^2)-t^2$, then 
we can move the integration with respect to $t$ to a contour $L_t=L_-\cup L_0\cup L_+\subset\{z\colon \Im z\ge |\Re z|\}$ symmetric with respect to the imaginary axis, satisfying the same properties as in Subsection~\ref{sbsc:<ge^nF>} but for $\ep=0$. 

We also change $r$-contour as follows:
\begin{equation}\label{R_cont_l:uc_int}
\begin{split}
&\wt{\mathcal{R}}(t)=\mathcal{R}_1(t)\cup\mathcal{R}_2(t),\quad \text{where\quad}
\mathcal{R}_1(t)=[it,0],\quad
\mathcal{R}_2(t)=\{\rho\colon \rho\ge 0\}.
\end{split}
\end{equation}
For $(u,t,s,r)$ lying in a neighbourhood of $(x_{0,n},ix_{\ep,n},0,0)$ we have
\begin{equation*}
\begin{split}
&F_{n,1}(u)=\L_\n(x_{0,\n}^2)-x_{0,\n}^2-K_1 (u-x_{0,\n})^2+O(|u-x_{0,\n}|^3),\\
&F_{n,2}(t,s,r)
=-\L_n(x_{0,\n}^2)+x_{0,\n}^2-K_1(t-ix_{0,\n})^2-s^2-r^2+O(|t-x_{0,\n}|^3+|s|^3),
\end{split}
\end{equation*}
where
\begin{equation}\label{def_K_1}
K_1=2x_{0,\n}^2\cdot\tr G^2(x_{0,\n}^2).
\end{equation}
Since $K_1\ge c$ uniformly, the following inequalities hold:
\begin{equation}\label{F_1,2_ineq_rstr_uc_1}
\begin{split}
nF_{n,1}(u)\le nF_{n,1}(x_{0,\n})-c\log^2 n   \text{\quad when }|u-x_{0,\n}|>&n^{-1/2}\log n,\\
n\Re F_{n,2}(t,s,r)\le n F_{n,2}(ix_{0,\n},0,0)-c\log^2 n\text{\quad when\quad} &\max\{|s|,|t-ix_{0,\n}|\}>n^{-1/2}\log n\\
 &\text{ or } r\in\mathcal{R}_2(t),\ |r|>n^{-1/2}\log n
\end{split}
\end{equation}
The estimations \eqref{F_1,2_ineq_rstr_uc_1} allow us to restrict the integration to the neighbourhood of $(x_{0,n},ix_{0,n},0,0)$. Make the change of variables:
$
u=x_{0,\n}+\wt{u}n^{-1/2},\quad t=ix_{0,\n}+\wt{t}n^{-1/2},\quad s=\wt{s}n^{-1/2},\quad r=\wt{r}n^{-1/2}.
$
Expanding $\wt{\varphi}$ up to the first order and using Remark~\ref{rem:phi_bounds} one can obtain
\begin{equation*}
|\wt{\varphi}(x_{0,\n}+\wt{u}n^{-1/2},ix_{0,\n}+\wt{t}n^{-1/2},\wt{s}n^{-1/2})|\le 
C(|\wt{u}|n^{-1/2}+|\wt{t}|n^{-1/2})+O(n^{-1}\log^2 n),
\end{equation*}
and the other multipliers of $\Phi_n$ are bounded in the neighbourhood. We also estimate 
$$
\left|\int_{\mathcal{R}_{1+}(t)} e^{-nr^2}\,dr\right|\le \dfrac{Cn^{-1/2}}{\sqrt{\Re(\tfrac{-t^2}{|t|^2})}}+Cn^{-1/2}\le Cn^{-1/2}.
$$
According to the outline of the proof, it is left to check that $\dfrac{n^{5/2}}{n^{3\cdot 1/2}}\cdot Cn^{-1/2}\cdot Cn^{-1/2}$ is bounded, which is true.

\end{proof}

A similar result holds for $z\in E$ lying outside of $D$ far enough from $\partial D$. More precisely,

\begin{pro}\label{pro:uni_bound_far_out_sm}
Fix an arbitrary $d>0$, and set $E_{d,out}:=\{z\in E\mid z\notin D,\ \dist(z,\partial D)\ge d\}$. Then there exist $n_0\in\mathbb{N}$ and $C>0$ such that 
$$
\ep^{1/2}\,|T(\ep,n,z)|\le C
$$ 
for $n\ge n_0$, $0<\ep<n^{-1}$ and $z\in E_{d,out}$.
\end{pro}

\begin{proof}
Since $E_{d,out}$ is a compact subset of $\mathbb{C}\setminus\overline{D}$, we can use the results of Subsection~\ref{sbsc:saddle} for $E_{out}=E_{d,out}$.

Similarly to the proof of Proposition~\ref{pro:uni_bound_far_in_sm} we obtain \eqref{T_repr_uc_1}. Notice that $u=0$ is the maximum point of $F_{n,1}(u)$, $u\in[0,+\infty)$ since $F_{n,1}'(u)=2u(\tr G(u^2)-1)<0$ for $u>0$. For $F_{n,2}$,
change the $t$-contour and $r$-contour as follows:
\begin{equation}\label{r,t_uni_out_sm}
\begin{split}
&\mathbf{L}_t=\mathbf{L}_1\cup \mathbf{L}_2,\ \text{where }
\mathbf{L}_{1,2}=\{((\pm1+i)\tau, \tau\ge 0\};\\
&\wt{\mathcal{R}}(t)=\mathcal{R}_1(t)\cup\mathcal{R}_2(t),\quad\text{where\quad}
\mathcal{R}_1(t)=[it,0],\
\mathcal{R}_2(t)=\{r\colon r\ge 0\}.
\end{split}
\end{equation} 
One can check that for small $u,s,t,r$ we have
\begin{equation*}
\begin{split}
&F_{n,1}(u)=\L_n(0)-K_2u^2+O(u^4);\\
&F_{n,2}(t,s,r)=-\L_n(0)-iK_2|t|^2-K_4|t|^4-K_3s^2-r^2+O(|s|^3+|t|^6);
\end{split}
\end{equation*}
where
\begin{equation}\label{def_K2_K3_K4}
K_2=1-\tr G(0),\quad K_3=\tr G(0),\quad K_4=\dfrac{1}{2}\tr G^2(0).
\end{equation}
Obviously, $K_2,K_3,K_4\ge c>0$ uniformly in $n$, hence the following inequalities hold:
\begin{equation*}
\begin{split}
nF_{n,1}(u)\le nF_{n,1}(0)-c\log^2 n, \quad\text{when } &|u|>n^{-1/2}\log n;\\
n\Re F_{n,2}(t,s,r)\le nF_{n,2}(0,0,0)-c\log^2 n,\quad &\text{when } |s|>n^{-1/2}\log n\text{ or }|t|>n^{-1/4}\log^{1/2}n\\
&\text{ or } r\in\mathcal{R}_2(t),\ |r|>n^{-1/2}\log n.
\end{split}
\end{equation*}
Then we can restrict the integration to the neighbourhood of $(0,0,0,0)$.
Make a change of variables $u=n^{-1/2}\wt{u}$,  $t=n^{-1/4}(\pm1+i)\wt{t}_{\pm}$, $s=n^{-1/2}\wt{s}$. Observe that 
$$
u^2I_1(2u)+\frac{1}{n}uI_0(2u)=O(n^{-3/2}\log^3 n)=O(n^{-1}),
$$ 
and the other multipliers of $\Phi_n$ are bounded. 
Finally, we estimate
$$
\int_{\mathcal{R}_{1+}(t)} e^{-nr^2}\,dr\le C(|t|+n^{-1/2})\le Cn^{-1/4}(1+|\wt{t}_{\pm}|)
$$ 
for $t$ in the neighbourhood.
According to the outline of the proof, we are left to check that \linebreak $\dfrac{n^{5/2}}{n^{1/2+1/2+1/4}}\cdot Cn^{-1}\cdot Cn^{-1/4}$ is bounded, which is true.

\end{proof}

Next, we obtain a bound on $T(\ep,n,z)$  when $z$ is close to the boundary $\partial D$ of $D$.

\begin{pro}\label{pro:uni_bound_close_small}
There exist $d>0$, $n_0\in\mathbb{N}$ and $C>0$ such that 
$$
\ep^{3/4}\,|T(\ep,n,z)|\le C
$$ 
for $n\ge n_0$, $0<\ep<n^{-1}$ and $z\in E_{d,\partial}$, where $E_{d,\partial}=\{z\in E\mid \dist(z,\partial D)\le d\}$.
\end{pro}

\begin{proof}
Similarly to Proposition~\ref{pro:uni_bound_far_in_sm} we have
\begin{equation*}
\ep^{3/4}\,|T(\ep,n,z)|\le Cn^{9/4}\int_{\mathcal{V}}\Phi_n(\mathbf{x})\,e^{n\Re F_n(\mathbf{x})}\,d\mathbf{x},
\end{equation*}
with the same notations as in \eqref{T_repr_uc_1}. We split the proof into four cases:\bigskip

\noindent
\textbf{Case 1.} $1-\tr G(0)\ge n^{-1/2}$. In this case change $t$-contour and $r$-contour as in~\eqref{r,t_uni_out_sm}, then for small $u,t,s,r$ we have
\begin{equation*}
\begin{split}
&F_{n,1}(u)=\L_n(0)-K_2u^2+O(u^4);\\
&F_{n,2}(t,s,r)=-\L_n(0)-iK_2|t|^2-K_4|t|^4-K_3s^2-r^2+O(|s|^3+|t|^6);
\end{split}
\end{equation*}
where $K_{2,3,4}$ are defined in \eqref{def_K2_K3_K4}. In this case we have $K_2\ge n^{-1/2}$, $K_3,K_4\ge c>0$. Thus the following inequalities hold:
\begin{equation*}
\begin{split}
nF_{n,1}(u)\le nF_{n,1}(0)-c\log^2 (K_2n), \quad\text{when } |u|&>(K_2n)^{-1/2}\log (K_2n);\\
n\Re F_{n,2}(t,s,r)\le nF_{n,2}(0,0,0)-c\log^2 n,\quad \text{when }& |s|>n^{-1/2}\log n\text{ or }|t|>n^{-1/4}\log^{1/2}n\\
&\text{ or } r\in\mathcal{R}_2(t), |r|>n^{-1/2}\log n.
\end{split}
\end{equation*}
Since $K_2n\ge n^{1/2}$, the above bounds allow us to restrict the integration to the neighbourhood of the saddle point $(0,0,0,0)$. Make a change $u=(K_2n)^{-1/2}\wt{u}$, $t=n^{-1/4}(\pm1+i)\wt{t}_{\pm}$, $s=n^{-1/2}\wt{s}$ and observe that $\wt\varphi$ is even with respect to $\wt u$, $\wt t$, $\wt s$. Then we can estimate the multipliers in $\Phi_n(\mathbf{x})$ as follows: $\sqrt{|r-it|}=O(1)$,
\begin{equation*}
\begin{split}
&|\wt\varphi(u,t,s)|\le |\wt\varphi(0,0,0)|+C\Bigl(\dfrac{|\wt{t}_{\pm}|^2}{n^{1/2}}+\dfrac{|\wt{s}|^2}{n}+\dfrac{|\wt{u}|^2}{K_2n}\Bigr)\le K_2^2+n^{-1/2}\mathcal{P}(\wt{u},\wt{t}_{\pm},\wt{s})\le n^{1/4}K_2^{3/2}\mathcal{P}(\wt{u},\wt{t}_{\pm},\wt{s}),\\
&|u^2I_1(2a u)+\dfrac{1}{n}uI_0(2 au)|\le C((K_2n)^{-1}|\wt{u}|^2+n^{-1})\le C(K_2n)^{-1}(1+|\wt{u}|^2),
\end{split}
\end{equation*}
since $K_2\ge n^{-1/2}$. We obtain $\Phi_n(\mathbf{x})\le CK_2^{1/2}n^{-3/4}\mathcal{P}(\wt{u},\wt{t}_{\pm},\wt{s})$.
We also estimate 
$$
\Big|\int\limits_{\mathcal{R}_{1+}(t)} e^{-nr^2}\,dr\Big|\le |t|+Cn^{-1/2}\le Cn^{-1/4}(1+|\wt{t}_{\pm}|)
$$ 
for $t$ in the neighbourhood. We are left to check that $\dfrac{n^{9/4}}{n^{1/4+1/2}\, (K_2n)^{1/2}}\cdot CK_2^{1/2}n^{-3/4}\cdot Cn^{-1/4}$ is bounded, which is true.
\bigskip

\noindent
\textbf{Case 2.} $0\le 1-\tr G(0)\le n^{-1/2}$. In this case change $t$-contour and $r$-contour as in~\eqref{r,t_uni_out_sm}, then for small $u,t,s,r$ we have
\begin{equation*}
\begin{split}
&F_{n,1}(u)=\L_n(0)-K_2u^2-K_4u^4+O(u^6)\le \L_n(0)-K_4u^4+O(u^6);\\
&F_{n,2}(t,s,r)=-\L_n(0)-iK_2|t|^2-K_4|t|^4-K_3s^2-r^2+O(|s|^3+|t|^6);
\end{split}
\end{equation*}
Since $K_{3,4}\ge C>0$, the following inequalities hold:
\begin{equation*}
\begin{split}
nF_{n,1}(u)\le nF_{n,1}(0)-c\log^2 n, \quad\text{when } u & > n^{-1/4}\log^{1/2}n;\\
n\Re F_{n,2}(t,s,r)\le nF_{n,2}(0,0,0)-c\log^2 n,\quad &\text{when } |s|>n^{-1/2}\log n\text{ or }|t|>n^{-1/4}\log^{1/2}n\\
&\text{ or } r\in\mathcal{R}_2(t), |r|>n^{-1/2}\log n.
\end{split}
\end{equation*}
We can restrict the integration to the neighbourhood of $(0,0,0,0)$, make a change $u=n^{-1/4}\wt{u}$, $t=n^{-1/4}(\pm1+i)\wt{t}_{\pm}$, $s=n^{-1/2}\wt{s}$ and estimate the multipliers in $\Phi_n(\mathbf{x})$ as follows: $\sqrt{|r-it|}=O(1)$,
\begin{equation*}
\begin{split}
&|\wt\varphi(u,t,s)|\le K_2^2+n^{-1/2}\mathcal{P}(\wt{u},\wt{t}_{\pm},\wt{s})\le n^{-1/2}\mathcal{P}(\wt{u},\wt{t}_{\pm},\wt{s}),\\
&|u^2I_1(2 u)+\dfrac{1}{n}uI_0(2 u)|\le C(n^{-1/2}|\wt{u}|^2+n^{-1})\le Cn^{-1/2}|\wt{u}|^2.
\end{split}
\end{equation*}
Hence, $\Phi_n(\wt{\mathbf{x}})\le Cn^{-1}\mathcal{P}(\wt{u},\wt{t}_{\pm},\wt{s})$. 
As in Case 1, $\Big|\int\limits_{\mathcal{R}_{1+}(t)} e^{-nr^2}\,dr\Big| \le Cn^{-1/4}(1+|\wt{t}_{\pm}|)$ for $t$ in the neighbourhood. We are left to check that $\dfrac{n^{9/4}}{n^{1/4+1/4+1/2}}\cdot Cn^{-1}\cdot Cn^{-1/4}$ is bounded, which is true.
\bigskip

\noindent
\textbf{Case 3.} $\tr G(0)> 1$, $x_{0,n}\ge n^{-1/4}$. Here we use the fact that for $\tr G(0)> 1$ the equation~\eqref{eq_x_0n} has exactly one positive root $x_{0,n}$.  Set $h_n(t)=-\L_n(-t^2)-t^2$.

We cannot shift $t$-contour the same way as in Subsection~\ref{sbsc:<ge^nF>} any more, since we want to keep everything uniform in $n$. If we choose similar $L_0=[ix_{0,n}-\delta,ix_{0,n}+\delta]$, then $\delta$ should be bounded as $\delta\le x_{0,n}$ in order for the contour to lie inside $\{z\colon \Im z\ge |\Re z|\}$, which shows that $\delta$ (and thus $\sigma$) cannot be independent of $n$. Instead, we choose a contour $L_t=L_-\cup L_0\cup L_+\subset\{z\colon \Im z\ge |\Re z|\}$ symmetric with respect to the imaginary axis, such that 
$$
L_0=[ix_{0,n}+\delta(-1+i(1-x_{0,n}));ix_{\ep,n}]\cup [ix_{0,n};ix_{0,n}+\delta(1+i(1-x_{0,n}))],
$$ 
$\Re h_n(t)$ decreases on $[ix_{0,n},ix_{0,n}+\delta(1+i(1-x_{0,n}))]$ and $\Re h_n(t)\le h_n(ix_{0,n})-\sigma$ for $t\in L_{\pm}$.
We can choose such $\delta,\sigma>0$  independent of $n$, since for $\wt h_n(\tau)=\Re h_n(ix_{0,n}+\tau(1+i(1-x_{0,n})))$
we have $\wt h''_n(0)<0$, $\wt h'''_n(0)<0$, $\wt h^{(IV)}_n(0)\le -c<0$ for small enough $x_{0,n}$ and $|\wt h^{(V)}_n(0)|$ is bounded, while there is also a restriction $\delta\le 1$ which is now also independent of $n$. 

We also change $r$-contour as follows:
\begin{equation*}
\begin{split}
\wt{\mathcal{R}}(t)&=\mathcal{R}_1(t)\cup\mathcal{R}_2(t),\quad \text{where\quad}
\mathcal{R}_1(t)=[it,0],\quad
\mathcal{R}_2(t)=\{r\colon r\ge 0\}.
\end{split}
\end{equation*}
For  $(u,t,s,r)$ lying in a neighbourhood of $(x_{0,n},ix_{0,n},0,0)$ and $t=ix_{0,n}+(\pm1+i(1-x_{0,n}))\tau$ we have
\begin{equation*}
\begin{split}
F_{n,1}(u)&=\L_n(x_{0,n}^2)-x_{0,n}^2-K_1(u-x_{0,n})^2+O(u^3);\\
F_{n,2}(t,s,r)&=-\L_n(x_{0,n}^2)+x_{0,n}^2-K_1(t-ix_{0,n})^2-s^2-r^2+O(|s|^3+|t-ix_{0,n}|^3)=\\
&=-\L_n(x_{0,n}^2)+x_{0,n}^2-K_1(2x_{0,n}-x_{0,n}^2+i(\ldots))\tau^2-s^2-r^2+O(|s|^3+\tau^3)
\end{split}
\end{equation*}
where $K_1$ is defined in \eqref{def_K_1}. Also, since $K_1\ge cx_{0,n}^2$ and $\Re((2x_{0,n}-x_{0,n}^2+i(\ldots))\ge cx_{0,n}$ for small enough $d$, the following inequalities hold:
\begin{equation*}
\begin{split}
&nF_{n,1}(u)\le nF_{n,1}(x_{0,n}) -c\log^2 (nx_{0,n}^2), \quad\text{when } |u-x_{0,n}|>(nx_{0,n}^2)^{-1/2}\log (nx_{0,n}^2);\\
&n\Re F_{n,2}(t,s,r)\le nF_{n,2}(ix_{0,n},0,0) -c\log^2 (nx_{0,n}^3),\quad \text{when } |t-ix_{0,n}|>(nx_{0,n}^3)^{-1/2}\log (nx_{0,n}^3);\\
&n\Re F_{n,2}(t,s,r)\le nF_{n,2}(ix_{0,n},0,0) -c\log^2 n,\quad \text{when } |s|>n^{-1/2}\log n\text{ or } r\in\mathcal{R}_2(t),\,|r|>n^{-1/2}\log n.
\end{split}
\end{equation*}
The inequalities $nx_{0,n}^2\ge n^{1/2}$, $nx_{0,n}^3\ge n^{1/4}$ imply that we can restrict the integration to the neighbourhood of the saddle point $u=x_{0,n}$, $t=ix_{0,n}$, $s=0$, $r=0$ and make a change $u=x_{0,n}+(nx_{0,n}^2)^{-1/2}\wt{u}$, $t=ix_{0,n}+(nx_{0,n}^3)^{-1/2}(\pm1+i(1-x_{0,n}))\wt{t}_{\pm}$, $s=n^{-1/2}\wt{s}$. Next we estimate the multipliers in $\Phi_n(\mathbf{x})$ in the neighbourhood. Obviously, $\sqrt{|r-it|}\le C$. Expanding $\wt{\varphi}$ up to the second order and using Remark~\ref{rem:phi_bounds} one can obtain
\begin{equation*}
\begin{split}
|\wt\varphi(u,t,s)|&\le Cx_{0,n}^3\Bigl(\dfrac{|\wt{u}|}{(nx_{0,n}^2)^{1/2}}+\dfrac{|\wt{t}_{\pm}|}{(nx_{0,n}^3)^{1/2}}\Bigr)+C\Bigl(\dfrac{|\wt{u}|^2}{nx_{0,n}^2}+\dfrac{|\wt{t}_{\pm}|^2}{nx_{0,n}^3}+\dfrac{|\wt{s}|^2}{n}\Bigr)\le x_{0,n}^{-1/2}n^{-3/8}\mathcal{P}(\wt{u},\wt{t}_{\pm},\wt{s}).
\end{split}
\end{equation*}
Since $I_1(2u)\le Cu$, $I_0(2u)\le C$ for $u$ in a neighbourhood, we have
$$
u^2I_1(2 u)+\dfrac{1}{n}uI_0(2 u)\le  Cx_{0,n}^3(1+|\wt{u}|^3).
$$
This shows that $\Phi_n(\wt{\mathbf{x}})\le x_{0,n}^{5/2}n^{-3/8}\,\mathcal{P}(\wt{u},\wt{t}_{\pm},\wt{s})$. Next we estimate integral with respect to $r$:
\begin{equation*}
\left|\int_{\mathcal{R}_{1+}(t)} e^{-nr^2}\,dr\right|\le \dfrac{Cn^{-1/2}}{\sqrt{\Re(\tfrac{-t^2}{|t|^2})}}+Cn^{-1/2}\le Cn^{-1/2}\Big(1+\frac{(nx_{0,n}^{3})^{-1/2}\wt{t}_{\pm}}{x_{0,n}}\Big)\le Cn^{-3/8}(1+|\wt{t}_{\pm}|)
\end{equation*}
since $x_{0,n}\ge n^{-1/4}$.
We are left to check that $\dfrac{n^{9/4}}{n^{1/2}(nx_{0,n})^{1/2}(nx_{0,n}^3)^{1/2}}\cdot x_{0,n}^{5/2}n^{-3/8}\cdot Cn^{-3/8}$ is bounded, which is true.
\bigskip

\noindent
\textbf{Case 4.} $\tr G(0)> 1$, $x_{0,n}\le n^{-1/4}$. We change $t$-contour and $r$-contour as in case 3. Denote $a_k:=\tr G^k(x_{0,n}^2)$.  Observe that $a_k\ge c>0$ for each $k$ uniformly in $n$ and
\begin{equation*}
\begin{split}
F_{n,1}(u)=&\L_n(u^2)-u^2=\L_n(x_{0,n}^2)-x_{0,n}^2-2a_2x_{0,n}^2(u-x_{0,n})^2-(2a_2x_{0,n}-\tfrac{8}{3}a_3x_{0,n}^3)(u-x_{0,n})^3-\\
&-(\tfrac{1}{2}a_2+O(x_{0,n}))(u-x_{0,n})^4+O((u-x_{0,n})^5).
\end{split}
\end{equation*}
One can easily check that for small enough $x_{0,n}$ and $u\ge 0$ we have
\begin{equation*}
\begin{split}
F_{n,1}(u)\le \L_n(x_{0,n}^2)-x_{0,n}^2-(\tfrac{1}{2}a_2+O(x_{0,n}))(u-x_{0,n})^4+O((u-x_{0,n})^5).
\end{split}
\end{equation*}
For $t=ix_{0,n}+(\pm1+i(1-x_{0,n}))\tau$ we have
\begin{equation*}
\begin{split}
\Re F_{n,2}(t,s,r)=&\L_n(x_{0,n}^2)-x_{0,n}^2-2a_2x_{0,n}^2(2x_{0,n}+O(x_{0,n}^2))\tau^2-(4a_2x_{0,n}+O(x_{0,n}^2))\tau^3-\\
&-(2a_2+O(x_{0,n}))\tau^4-\wt s^2-\wt r^2+O(\tau^5+|\wt s|^3)\le\\
\le & \L_n(x_{0,n}^2)-x_{0,n}^2-(2a_2+O(x_{0,n}))\tau^4-\wt s^2-\wt r^2+O(\tau^5+|\wt s|^3).
\end{split}
\end{equation*}
Since $a_2\ge c>0$, then the following inequalities hold:
\begin{equation*}
\begin{split}
nF_{n,1}(u)\le nF_{n,1}(x_{0,n}) -c\log^2 n, \quad\text{when } |u-&ix_{0,n}|>n^{-1/4}\log^{1/2}n;\\
n\Re F_{n,2}(t,s,r)\le nF_{n,2}(ix_{0,n},0,0) -c\log^2 n,\quad &\text{when } |s|>n^{-1/2}\log n\text{ or }|t-ix_{0,n}|>n^{-1/4}\log^{1/2}n\\
&\text{ or } r\in\mathcal{R}_2(t),\,|r|>n^{-1/2}\log n.
\end{split}
\end{equation*}
We can restrict the integration to the neighbourhood of the saddle point $u=x_{0,n}, t=ix_{0,n}, s=0,r=0$, make a change $u=x_{0,n}+n^{-1/4}\wt{u}$, $t=ix_{0,n}+n^{-1/4}(\pm1+i(1-x_{0,n}))\wt{t}_{\pm}$, $s=n^{-1/2}\wt{s}$ and estimate the multipliers in $\Phi_n(\mathbf{x})$ as follows: $\sqrt{|r-it|}\le C$;
\begin{equation*}
\begin{split}
&|\wt{\varphi}(u,s,t)|\le C(n^{-1/4}|\wt{u}|+n^{-1/4}|\wt{t}_{\pm}|+n^{-1}\wt{s}^2)\le Cn^{-1/4}\mathcal{P}(\wt{u},\wt{t}_{\pm},\wt{s});\\
&u^2I_1(2 u)+\dfrac{1}{n}uI_0(2 u)\le C(x_{0,n}+n^{-1/4}|\wt{u}|)^3+\frac{C}{n}(x_{0,n}+n^{-1/4}|\wt{u}|)\le Cn^{-3/4}\mathcal{P}(\wt{u}).
\end{split}
\end{equation*}
Hence, $\Phi_n(\wt{\mathbf{x}})\le Cn^{-1}\mathcal{P}(\wt{u},\wt{t}_{\pm},\wt{s})$.
Next we estimate 
$$
\Big|\int_{\mathcal{R}_{1+}(t)} e^{-nr^2}\,dr\Big|\le |t|+Cn^{-1/2}\le Cn^{-1/4}(1+|\wt{t}_{\pm}|)
$$
since $x_{0,n}\le n^{-1/4}$. We are left to check that $\dfrac{n^{9/4}}{n^{1/4+1/4+1/2}}\cdot Cn^{-1}\cdot Cn^{-1/4}$ is bounded, which is true.

\end{proof}

\subsubsection{Bounds on \boldmath{$T(\ep,n,z)$} for \boldmath{$\ep>n^{-1}$}}

Let $\ep>n^{-1}$. We start with the case when $z\in D$ far enough from $\partial D$.

\begin{pro}\label{pro:uni_bound_far_in_lar}
Fix an arbitrary $d>0$, and set $E_{d,in}:=\{z\in E\mid z\in D,\ \dist(z,\partial D)\ge d\}$. Then there exist $n_0\in\mathbb{N}$ and $C>0$ such that 
$$
\ep^{1/2}\,|T(\ep,n,z)|\le C
$$ 
for $n\ge n_0$, $n^{-1}<\ep<\ep_0$ and $z\in E_{d,in}$.
\end{pro}

\begin{proof}
Since $E_{d,in}$ is a compact subset of $\Int D$, we can use the results of Subsection~\ref{sbsc:saddle} for $E_{in}=E_{d,in}$.  

Make the same change of variables as in Subsection~\ref{sbsc:<ge^nF>}: $r=R+it+\ep\in\mathcal{R}(t)$, where $\mathcal{R}(t)=\{-it-\ep+\tau,\ \tau>0\}$, $u=\sqrt{u_1^2+u_2^2}\in[0,+\infty)$, $w=\sqrt{u_1+u}\in[0,\sqrt{2u}]$. 
Change $t$-contour as in Subsection~\ref{sbsc:<ge^nF>} and choose the following $r$-contour:
$
\wt{\mathcal{R}}(t)=\mathcal{R}_1(t)\cup\mathcal{R}_2(t),\quad \text{where } \mathcal{R}_1(t)=[it+\ep,0],\ \mathcal{R}_2(t)=\{r:\ r\ge 0\}.
$
We can make the following estimations, using the fact that $\Re (r-it-\ep)\ge 0$:
\begin{equation*}
\Bigl|\dfrac{r-it-\ep}{\sqrt{v^2+4(r-it-\ep)}}\Bigr|\le \sqrt{|r-it-\ep|};\quad
\dfrac{u}{\sqrt{2u-w^2}}\cdot|u-w^2-\ep|\le u(u+\ep)\cdot\dfrac{1}{\sqrt{2u-w^2}}.
\end{equation*}
Now we can integrate with respect to $v$, $w$. One can show that
$$
\int e^{-n\ep v^2}\,dv=\sqrt{\dfrac{\pi}{n\ep}},\quad
\int_0^{\sqrt{2u}} \dfrac{e^{-2n\ep w^2}}{\sqrt{2u-w^2}}\,dw=\frac{\pi}{2}e^{-2n\ep u}I_0(2n\ep u),
$$
where $I_0(x)$ is a modified Bessel function. Asymptotic formulas for $I_0(x)$ imply that $e^{-x}I_0(x)\le \dfrac{C}{\sqrt{x}}$ for all $x\ge 0$ and some $C>0$. We can apply the inequality $e^{-2n\ep u}I_0(2n\ep u)\le \dfrac{C}{\sqrt{2n\ep u}}$ to obtain
\begin{equation}\label{uni_int_lar_notat_1}
\begin{split}
&\ep^{1/2}\,|T(\ep,n,z)|
\le Cn^2\ep^{-1/2}\int_\mathcal{V} \Phi_n(\mathbf{x}) \exp\{n\,F_n(\mathbf{x})\}\,d\mathbf{x},\\
&\text{where}\quad 
\mathbf{x}=(u,t,s,r),\quad \mathcal{V}=[0,+\infty)\times L_t\times\mathbb{R}\times  {\mathcal{R}}(t),\quad F_n(\mathbf{x})=F_{n,1}(u)+F_{n,2}(t,s,r),\\
&F_{n,1}(u)=\L_n(u^2)-(u-\ep)^2,\quad F_{n,2}(t,s,r)=-\L_n(s^2-t^2)-(t-i\ep)^2-r^2,\\
&\Phi_n(\mathbf{x})=\sqrt{|r-it-\ep|}\cdot\sqrt{u}(u+\ep)\cdot |\wt{\varphi}(u,t,s)|,\quad
\wt{\varphi}(u,t,s)=\varphi(u^2,s^2-t^2,z,z)
\end{split}
\end{equation}
Expand $F_{n,1}$ and $F_{n,2}$ as follows:
\begin{equation}\label{F_1,2_expand_uc_2}
\begin{split}
&F_{n,1}(u)=\L_n(x_{\ep,n}^2)-(x_{\ep,n}-\ep)^2-\kappa_1(u-x_{\ep,n})^2+O\bigl(|u-x_{\ep,n}|^3\bigr),\\
&F_{n,2}(t,s,r)
=-\L_n(x_{\ep,n}^2)+(x_{\ep,n}-\ep)^2-\kappa_1 (t-ix_{\ep,n})^2-\kappa_2 s^2-r^2+O(|t-ix_{\ep,n}|^3+|s|^3),
\end{split}
\end{equation}
where $\kappa_{1,2}$ are defined in \eqref{def_k_1}, \eqref{def_k_2}. Proposition~\ref{pro:sols_bound_int} shows that $\kappa_1,\kappa_2\ge c>0$, then we can make the following estimations:
\begin{equation}\label{F_1,2_ineq_rstr_uc_2}
\begin{split}
nF_{n,1}(u)\le nF_{n,1}(x_{\ep,n})-c \log^2 n   \text{\quad when }&|u-x_{\ep,n}|>n^{-1/2}\log n;\\
n\Re F_{n,2}(t,s,r)\le  nF_{n,2}(ix_{\ep,n},0,0)-c\log^2 n &\text{\quad when } \max\{|s|,|t-ix_{\ep,n}|\}>n^{-1/2}\log n\\
&\text{ or } r\in\mathcal{R}_2(t), |r|>n^{-1/2}\log n
\end{split}
\end{equation}
Hence, we can restrict the integration to the neighbourhood of the saddle point $u=x_{\ep,n}, t=ix_{\ep,n},s=0,r=0$.
Make a change of variables 
$u=x_{\ep,n}+\wt{u}_1n^{-1/2}$, $t=ix_{\ep,n}+\wt{t}n^{-1/2}$, $
s=\wt{s}n^{-1/2}$.
Remark~\ref{rem:phi_bounds} implies that $\wt\varphi(x_{\ep,n},ix_{\ep,n},0)=O(\ep+n^{-1})$ since $c\le x_{\ep,n}\le C$. Hence,
$$
|\wt\varphi(u,t,s)|\le C(\ep+n^{-1/2}(|\wt{u}|+|\wt{t}|))+O(n^{-1}\log^2 n).
$$ 
The other multipliers in $\Phi_n$ are bounded. We also estimate 
$$
\left|\int_{\mathcal{R}_{1+}(t)} e^{-nr^2}\,dr\right|\le \dfrac{Cn^{-1/2}}{\sqrt{\Re(\tfrac{-t^2}{|t|^2})}}+Cn^{-1/2}\le Cn^{-1/2}.
$$
According to the outline of the proof, we are left to check that 
$
\dfrac{n^{2}\ep^{-1/2}}{n^{3\cdot 1/2}} \cdot C(\ep+n^{-1/2})\cdot Cn^{-1/2}
$
is bounded, which is true since $\ep>n^{-1}$.

\end{proof}

Next we estimate $T(\ep,n,z)$ when $z$ is close to the boundary $\partial D$ of $D$.

\begin{pro}\label{pro:uni_bound_close_lar}
There exist $d>0$, $n_0\in\mathbb{N}$ and $C>0$ such that 
$$
\ep^{5/6}\,|T(\ep,n,z)|\le C
$$ 
for $n\ge n_0$, $n^{-1}<\ep<\ep_0$ and $z\in E_{d,\partial}$, where $E_{d,\partial}=\{z\in E\mid \dist(z,\partial D)\le d\}$.
\end{pro}

\begin{proof}
We split the proof into two cases:\bigskip

\noindent
\textbf{Case 1.} $\tr G(0)\ge 1$. Let $x_{\ep,n}$ be the positive root of~\eqref{eq_x_en}. According to Subsection~\ref{sbsc:saddle}, we have $c\ep^{1/3}\le x_{\ep,n}\le C$ for some $c,C>0$.

Similarly to Case 3 of Proposition~\ref{pro:uni_bound_close_small}, we cannot shift $t$-contour the same way as in Subsection~\ref{sbsc:<ge^nF>}, since we want to keep everything uniform in $\ep$, $n$. If we choose  $L_0=[ix_{\ep,n}-\delta,ix_{\ep,n}+\delta]$, then $\delta$ should be bounded as $\delta\le x_{\ep,n}-\ep$ in order for the contour to lie inside $\{z\colon \Im z\ge |\Re z|+\ep\}$, which shows that $\delta$ (and thus $\sigma$) cannot be independent of $\ep$. Instead, we choose a contour $L_t=L_-\cup L_0\cup L_+\subset\{z\colon \Im z\ge |\Re z|+\ep\}$ symmetric with respect to the imaginary axis, such that 
$$
L_0=[ix_{\ep,n}+\delta(-1+i(1+\ep-x_{\ep,n}));ix_{\ep,n}]\cup [ix_{\ep,n};ix_{\ep,n}+\delta(1+i(1+\ep-x_{\ep,n}))],
$$ 
$\Re h_n(t)$ decreases on $[ix_{\ep,n},ix_{\ep,n}+\delta(1+i(1+\ep-x_{\ep,n}))]$ and $\Re h_n(t)\le h_n(ix_{\ep,n})-\sigma$ for $t\in L_{\pm}$.
We can choose such $\delta,\sigma>0$  independent of $\ep$, $n$ for the same reason as in Case 3 of Proposition~\ref{pro:uni_bound_close_small}.

Choose the following $r$-contour:
$
\wt{\mathcal{R}}(t)=\mathcal{R}_1(t)\cup\mathcal{R}_2(t),\quad \text{where } \mathcal{R}_1(t)=[it+\ep,0],\ \mathcal{R}_2(t)=\{r:\ r\ge 0\}.
$
Similarly to Proposition~\ref{pro:uni_bound_far_in_lar} we obtain
\begin{equation}\label{ttttt}
\begin{split}
\ep^{5/6}\,|T(\ep,n,z)|
\le Cn^2\ep^{-1/6}\int_\mathcal{V} \Phi_n(\mathbf{x}) \exp\{n\, F_n(\mathbf{x})\}\,d\mathbf{x},
\end{split}
\end{equation}
with the same notations as in \eqref{uni_int_lar_notat_1}.
Consider the following subcases:\bigskip

\noindent
\textbf{\textit{Subcase 1a.}} Suppose $\ep^{1/5}\le x_{\ep,n}\le C$. 
For $(u,t,s,r)$ lying in a neighbourhood of $(x_{\ep,n}, ix_{\ep,n},0,0)$ and $t=ix_{\ep,n}+(\pm1+i(1+\ep-x_{\ep,n}))\tau$ we have
\begin{equation*}
\begin{split}
&F_{n,1}(u)=\L_n(x_{\ep,n}^2)-(x_{\ep,n}-\ep)^2-\kappa_1 (u-x_{\ep,n})^2+O((u-x_{\ep,n})^3);\\
&F_{n,2}(t,s,r)=-\L_n(x_{\ep,n}^2)+(x_{\ep,n}-\ep)^2-\kappa_3\tau^2-\kappa_4\tau^3-\kappa_5\tau^4-\kappa_2s^2-r^2+O(|s|^3+|t-ix_{\ep,n}|^5).
\end{split}
\end{equation*}
where $\kappa_1,\kappa_2$ are defined in \eqref{def_k_1}, \eqref{def_k_2},  
\begin{equation}\label{def_k_3}
\kappa_3=\kappa_1\cdot(\pm 2i+2(1\mp i)(x_{\ep,n}-\ep)+(x_{\ep,n}-\ep)^2).
\end{equation} 
and $\kappa_4,\kappa_5$ satisfy
\begin{equation}\label{def_k_4_k_5}
\Re\kappa_4=4x_{\ep,n}\cdot\tr G^2(x_{\ep,n}^2)+O(x_{\ep,n}^2),\quad 
\kappa_5=2\cdot\tr G^2(x_{\ep,n}^2)+O(x_{\ep,n}^2).
\end{equation}
It is easy to see that $\kappa_1\ge cx_{\ep,n}^2$, $\kappa_2\ge c$, $\Re\kappa_3\ge cx_{\ep,n}^3$. Thus the following inequalities hold:
\begin{equation*}
\begin{split}
&nF_{n,1}(u)\le nF_{n,1}(x_{\ep,n})-c\log^2 (nx_{\ep,n}^2), \quad\text{when } |u-x_{\ep,n}|>(nx_{\ep,n}^2)^{-1/2}\log (nx_{\ep,n}^2);\\
&n\Re F_{n,2}(t,s,r)\le nF_{n,2}(ix_{\ep,n},0,0)-c\log^2 (nx_{\ep,n}^3),\quad \text{when } |t-ix_{\ep,n}|>(nx_{\ep,n}^3)^{-1/2}\log(nx_{\ep,n}^3);\\
&n\Re F_{n,2}(t,s,r)\le nF_{n,2}(ix_{\ep,n},0,0)-c\log^2 n,\quad \text{when } |s|>n^{-1/2}\log n \text{ or } r\in\mathcal{R}_2(t), |r|>n^{-1/2}\log n.
\end{split}
\end{equation*}
Since $nx_{\ep,n}^2\ge n\ep^{1/2}\ge cn^{1/2}$ and $nx_{\ep,n}^3\ge n\ep^{3/4}\ge cn^{1/4}$, the above bounds allow us to restrict the integration to the neighbourhood of $(x_{\ep,n},ix_{\ep,n},0,0)$. Make a change $u=x_{\ep,n}+(nx_{\ep,n}^2)^{-1/2}\wt{u}$, $t=ix_{\ep,n}+(nx_{\ep,n}^3)^{-1/2}(\pm1+i(1+\ep-x_{\ep,n}))\wt{t}_{\pm}$, $s=n^{-1/2}\wt{s}$. Expanding $\wt{\varphi}$ up to the second order and using Remark~\ref{rem:phi_bounds} together with  $x_{\ep,n}\ge c\ep^{1/3}$ one can obtain
\begin{equation*}
\begin{split}
&|\wt\varphi(u,t,s)|\le C\Bigl(\ep x_{\ep,n}+\dfrac{x_{\ep,n}^3}{(nx_{\ep,n}^2)^{1/2}}|\wt{u}|+\dfrac{x_{\ep,n}^3}{(nx_{\ep,n}^3)^{1/2}}|\wt{t}_{\pm}|+\dfrac{x_{\ep,n}^3}{n^{1/2}}|\wt{s}|+\dfrac{1}{nx_{\ep,n}^2}|\wt{u}|^2+\dfrac{1}{nx_{\ep,n}^3}|\wt{t}_{\pm}|^2+\dfrac{1}{n}|\wt{s}|^2\Bigr)\le\\
&\le \ep^{1/6}x_{\ep,n}\cdot\mathcal{P}(\wt{u},\wt{t}_{\pm},\wt{s}).
\end{split}
\end{equation*}
Also we have $u\le Cx_{\ep,n}(1+|\wt{u}|)$, thus $\sqrt{u}(u+\ep)\le Cx_{\ep,n}^{3/2}(1+|\wt{u}|^2)$ and 
$$
{\Phi}_n(\wt{\mathbf{x}})\le \ep^{1/6}x_{\ep,n}^{5/2}\cdot\mathcal{P}(\wt{u},\wt{t}_{\pm},\wt{s}).
$$
Next, we estimate integral with respect to $r$:
\begin{equation*}
\left|\int_{\mathcal{R}_{1+}(t)} e^{-nr^2}\,dr\right|\le \dfrac{Cn^{-1/2}}{\sqrt{\Re(\tfrac{-t^2}{|t|^2})}}+Cn^{-1/2}\le Cn^{-1/2}\Big(1+\frac{(nx_{\ep,n}^{3})^{-1/2}\wt{t}_{\pm}}{x_{\ep,n}}\Big)\le Cn^{-1/2}\Big(1+|\wt{t}_{\pm}|\Big).
\end{equation*}
since $x_{\ep,n}\ge \ep^{1/5}$ and $\ep>n^{-1}$.
We are left to check that 
$
\dfrac{n^2\ep^{-1/6}}{(nx_{\ep,n}^2)^{1/2}(nx_{\ep,n}^3)^{1/2}n^{1/2}}\cdot \ep^{1/6}x_{\ep,n}^{5/2}\cdot Cn^{-1/2}
$
is bounded, which is true.
\bigskip

\noindent
\textbf{\textit{Subcase 1b.}} Suppose that $\ep> n^{-1/2}$.
Then the inequality $x_{\ep,n}\ge c\ep^{1/3}> cn^{-1/6}$  shows that $nx_{\ep,n}^2>cn^{2/3}$ and $nx_{\ep,n}^3>cn^{1/2}$. Then we can restrict the integration to the same neighbourhood as in Subcase~1a and make the same change of variables. All the further bounds from the previous subcase still hold, and 
$
\dfrac{n^2\ep^{-1/6}}{(nx_{\ep,n}^2)^{1/2}(nx_{\ep,n}^3)^{1/2}n^{1/2}}\cdot \ep^{1/6}x_{\ep,n}^{5/2}\cdot Cn^{-1/2}
$
is still bounded.
\bigskip

\noindent
\textit{\textbf{Subcase 1c.}} Suppose $n^{-1}<\ep<n^{-1/2}$, $c\ep^{1/3}\le x_{\ep,n}\le \ep^{1/5}$. For $t=ix_{\ep,n}+(\pm1+i(1+\ep-x_{\ep,n}))\tau$ we have
\begin{equation*}
\begin{split}
\Re F_{n,2}(t,s,r)\le-\L_n(x_{\ep,n}^2)+(x_{\ep,n}-\ep)^2-\Re\kappa_5\tau^4-\kappa_2s^2-r^2+O(|s|^3+|t-ix_{\ep,n}|^5).
\end{split}
\end{equation*}
Since $\kappa_1\ge cx_{\ep,n}^2$, $\kappa_2,\Re\kappa_5\ge c>0$, then the following inequalities hold:
\begin{equation*}
\begin{split}
nF_{n,1}(u)\le nF_{n,1}(x_{\ep,n})-c\log^2 (nx_{\ep,n}^2), \quad\text{when } |u&\,-x_{\ep,n}|>(nx_{\ep,n}^2)^{-1/2}\log (nx_{\ep,n}^2);\\
n\Re F_{n,2}(t,s,r)\le nF_{n,2}(ix_{\ep,n},0,0)-c\log^2 n,\quad \text{when }& |s|>n^{-1/2}\log n\text{ or } |t-ix_{\ep,n}|>n^{-1/4}\log^{1/2}n\\
&\text{ or } r\in\mathcal{R}_2(t),\ |r|>n^{-1/2}\log n.
\end{split}
\end{equation*}
We have $nx_{\ep,n}^2\ge n\ep^{2/3}\ge cn^{1/3}$, thus the above bounds allow us to restrict the integration to the neighbourhood of the saddle point $(x_{\ep,n},ix_{\ep,n},0,0)$. Make a change $u=x_{\ep,n}+(nx_{\ep,n}^2)^{-1/2}\wt{u}$, $t=ix_{\ep,n}+n^{-1/4}(\pm1+i(1+\ep-x_{\ep,n}))\wt{t}_{\pm}$, $s=n^{-1/2}\wt{s}$. Expanding $\wt{\varphi}$ up to the fourth order and using Remark~\ref{rem:phi_bounds} one can obtain
\begin{equation*}
\begin{split}
|\wt\varphi(u,t,s)|&\le \Bigl(\ep x_{\ep,n}+\dfrac{x_{\ep,n}^3}{n^{1/4}}+\dfrac{x_{\ep,n}^3}{(nx_{\ep,n}^2)^{1/2}}+\dfrac{x_{\ep,n}^3}{n^{1/2}}+\dfrac{x_{\ep,n}^2}{n^{1/2}}+\dfrac{x_{\ep,n}^2}{nx_{\ep,n}^2}+\dfrac{x_{\ep,n}^2}{n}+\\
&+\dfrac{x_{\ep,n}}{n^{3/4}}+\dfrac{x_{\ep,n}}{(nx_{\ep,n}^2)^{3/2}}+\dfrac{x_{\ep,n}}{n^{3/2}}+\dfrac{1}{n}+\dfrac{1}{(nx_{\ep,n}^2)^{2}}+\dfrac{1}{n^2}\Bigr)\times\mathcal{P}(\wt{u},\wt{t}_{\pm},\wt{s})\le n^{-1/4}x_{\ep,n}\cdot\mathcal{P}(\wt{u},\wt{t}_{\pm},\wt{s}).
\end{split}
\end{equation*}
We also have 
$
|u|\le x_{\ep,n}+\dfrac{1}{n^{1/2}x_{\ep,n}}|\wt{u}|\le \ep^{1/6}(1+|\wt{u}|),
$
thus 
$\sqrt{u}(u+\ep)\le \ep^{1/4}(1+|\wt{u}|^2)$.
Then we can write
$
\Phi_n(\wt{\mathbf{x}})\le \ep^{1/4}n^{-1/4}x_{\ep,n}\cdot\mathcal{P}(\wt{u},\wt{t}_{\pm},\wt{s}).
$
Next, we estimate integral with respect to $r$:
\begin{equation*}
\left|\int_{\mathcal{R}_{1+}(t)} e^{-nr^2}\,dr\right|\le \dfrac{Cn^{-1/2}}{\sqrt{\Re(\tfrac{-t^2}{|t|^2})}}+Cn^{-1/2}\le Cn^{-1/2}\Big(1+\frac{n^{-1/4}\wt{t}_{\pm}}{x_{\ep,n}}\Big)\le Cn^{-1/2}\ep^{-1/12}\Big(1+|\wt{t}_{\pm}|\Big)
\end{equation*}
since $x_{\ep,n}\ge \ep^{1/3}$ and $n^{-1}<\ep$.
We are left to check that 
$
\dfrac{n^2\ep^{-1/6}}{n^{1/4+1/2}(nx_{\ep,n}^2)^{1/2}}\cdot \ep^{1/4}n^{-1/4}x_{\ep,n}\cdot Cn^{-1/2}\ep^{-1/12}
$
is bounded, which is true.
\bigskip

\noindent
\textbf{Case 2.}
$\tr G(0)< 1$. Let $x_{\ep,n}$ be the positive root of~\eqref{eq_x_en}. According to Proposition~\ref{pro:sols_bound_boundary}, we have $(1+c)\ep\le x_{\ep,n}\le C\ep^{1/3}$ for some $c,C>0$.

The method used in the previous cases does not give sufficient bound in this case. Instead, we introduce a trick similar to the one in Proposition~\ref{pro:for_I_1}. Namely, we construct an extra zero of the integrand at the saddle point with the use of identity~\eqref{<phi_exp_nF>=1}. 

Introduce the averaging:
\begin{equation}\label{def_<<>>}
\begin{split}
\llangle[\big] f(u_1,u_2,t,s,R,v)\rrangle[\big]
=&\dfrac{2n^3}{\pi^3}\int_0^\infty dR\int_{-\infty}^\infty dv\,du_1\,du_2\,ds\int_L dt\cdot \dfrac{R}{\sqrt{v^2+4R}}\cdot \varphi(u_1^2+u_2^2,s^2-t^2,z,z)\times\\
&\times f(u_1,u_2,t,s,R,v)\cdot \exp\left\{n\Bigl(\L_n(u_1^2+u_2^2)-(u_1+\ep)^2-u_2^2\Bigr)\right\}\times\\
&\times  \exp\left\{-n\Bigl(\L_n(s^2-t^2)+(t-i\ep)^2+(R+it+\ep)^2+\ep v^2\Bigr)\right\},
\end{split}
\end{equation}
Then 
$$
\ep^{5/6}\,T(\ep,n,z)=
-2\ep^{5/6}\llangle[\big] u_1+\ep\rrangle[\big]
=2\ep^{5/6}(x_{\ep,n}-\ep)\llangle[\big] 1\rrangle[\big]-2\ep^{5/6}\llangle[\big] u_1+x_{\ep,n}\rrangle[\big].
$$
Identities \eqref{I_rep_Z} and \eqref{<phi_exp_nF>=1} yield that $\llangle[\big] 1\rrangle[\big]=\mathcal{Z}(\ep,\ep,z,z)=1$, hence it suffices to prove that $\ep^{5/6}\big|\,\llangle[\big] u_1+x_{\ep,n}\rrangle[\big]\,\big|$ is bounded.  
Having $\ep^{5/6}\llangle[\big] u_1+x_{\ep,n}\rrangle[\big]$ one can change $t$-contour as in Case 1 and choose the following $r$-contour:
$
\wt{\mathcal{R}}(t)=\mathcal{R}_1(t)\cup\mathcal{R}_2(t),\quad \text{where } \mathcal{R}_1(t)=[it+\ep,0],\ \mathcal{R}_2(t)=\{r:\ r\ge 0\}.
$
For $t\in\mathbf{L}_t$, $r\in\wt{\mathcal{R}}(t)$ we have $\Bigl|\dfrac{r-it-\ep}{\sqrt{v^2+4(r-it-\ep)}}\Bigr|\le \dfrac{1}{2}\sqrt{|r-it-\ep|}$. Next we integrate with respect to $v$ and obtain
\begin{equation}\label{uni_notat_out}
\begin{split}
&\ep^{5/6}\bigl|\llangle[\big] u_1+x_{\ep,n}\rrangle[\big]\bigr|\,
\le Cn^{5/2}\ep^{1/3}\int_\mathcal{V} \wh\Phi_n(\mathbf{x}) \exp\{n\,\wh F_n(\mathbf{x})\}\,d\mathbf{x},\\
&\text{where } F_{n,1}(u_1,u_2)=\L_n(u_1^2+u_2^2)-(u_1+\ep)^2-u_2^2,\ F_{n,2}(t,s,r)=-\Bigl(\L_n(s^2-t^2)+(t-i\ep)^2+r^2\Bigr),\\
&\wh F_n(\mathbf{x})=F_{n,1}(u_1,u_2)+F_{n,2}(t,s,r),\quad
\wt\varphi(u_1,u_2,t,s)=\varphi(u_1^2+u_2^2,s^2-t^2,z,z),\\
&\wh\Phi_n(\mathbf{x})=|u_1+x_{\ep,n}|\,\sqrt{|r-it-\ep|}\cdot |\wt\varphi(u_1,u_2,t,s)|.
\end{split}
\end{equation} 
Consider the following subcases:\bigskip

\noindent
\textbf{\textit{Subcase 2a.}} Suppose $n^{-1}\le \ep\le n^{-1/3}$. 

We start with determining the maximum of $F_{n,1}(u_1,u_2)$. The maximum point $(u_1',u_2')$ is a solution of the following system:
\begin{equation*}
\partial_{u_1}F_{n,1}(u'_1,u'_2)=\partial_{u_2}F_{n,1}(u'_1,u'_2)=0,
\end{equation*}
or, equivalently, $u_2'=0$ and $(-u_1)'$ is a solution of \eqref{eq_x_en}. According to Proposition~\ref{pro:eq_sols_int} and Proposition~\ref{pro:eq_sols_ext}, the equation \eqref{eq_x_en} has exactly one positive root $x_{\ep,n}$ and no negative roots, thus $(u_1',u_2')=(-x_{\ep,n},0)$ is the maximum point of $F_{n,1}(u_1,u_2)$.

We have following expansions:
\begin{equation*}
\begin{split}
&F_{n,1}(u_1,u_2)=\L_n(x_{\ep,n}^2)-(x_{\ep,n}-\ep)^2-\kappa_1(u+x_{\ep,n})^2-\kappa_6u_2^2+O(|u_1+x_{\ep,n}|^3+|u_2|^3);\\
&F_{n,2}(t,s,r)=-\L_n(x_{\ep,n}^2)+(x_{\ep,n}-\ep)^2-(\kappa_3 \tau^2+\kappa_4\tau^3+\kappa_5 \tau^4+\kappa_2 s^2+r^2) +O(|s|^3+|\tau|^5),
\end{split}
\end{equation*}
where $t=ix_{\ep,n}+(\pm1+i(1+\ep-x_{\ep,n}))\tau$, $\kappa_1,\kappa_2,\kappa_3$ are defined in \eqref{def_k_1}, \eqref{def_k_2},  \eqref{def_k_3},
\begin{equation}\label{def_k_6}
\kappa_6=\dfrac{\ep}{x_{\ep,n}}
\end{equation}
and $\kappa_4,\kappa_5$ satisfy
\eqref{def_k_4_k_5}.
It is easy to see that $\kappa_{1,6}\ge c\ep/x_{\ep,n}$, $\kappa_2\ge c$, $\Re\kappa_{3,4}>0$, $\Re\kappa_5\ge c>0$.
Thus
\begin{equation*}
\begin{split}
&nF_{n,1}(u_1,u_2)\le nF_{n,1}(x_{\ep,n},0)-c\log^2 \dfrac{n\ep}{x_{\ep,n}}, \quad\text{when } \max\{|u_1+x_{\ep,n}|,|u_2|\}>(n\ep/x_{\ep,n})^{-1/2}\log (n\ep/x_{\ep,n}),\\
&n\Re F_{n,2}(t,s,r)\le nF_{n,2}(ix_{\ep,n},0,0)-c\log^2 n, \quad\text{when }|s|>n^{-1/2}\log n \text{ or } r\in\mathcal{R}_2(t), |r|>n^{-1/2}\log n,\\
&n\Re F_{n,2}(t,s,r)\le nF_{n,2}(ix_{\ep,n},0,0)-c\log^2 n, \quad\text{when } |t-ix_{\ep,n}|>n^{-1/4}\log^{1/2} n.
\end{split}
\end{equation*}
Since $n\ep/x_{\ep,n}\ge cn\ep^{2/3}\ge cn^{1/3}$, 
the above bounds allow us to restrict the integration to the neighbourhood of the saddle point $u_1=-x_{\ep,n}$, $u_2=0$, $t=ix_{\ep,n}$, $s=0$, $r=0$. Make a change $u_1=-x_{\ep,n}+(n\ep/x_{\ep,n})^{-1/2}\wt{u}_1$, $u_2=(n\ep/x_{\ep,n})^{-1/2}\wt{u}_2$, $t=ix_{\ep,n}+n^{-1/4}(\pm1+i(1+\ep-x_{\ep,n}))\wt{t}_{\pm}$, $s=n^{-1/2}\wt{s}$. Expanding $\wt{\varphi}$ up to the fourth order and using Remark~\ref{rem:phi_bounds} one can obtain
\begin{equation*}
\begin{split}
|\wt\varphi(u_1,u_2,t,s)|\le \Bigl(\ep^2 x_{\ep,n}^{-2}+\dfrac{\ep}{n^{1/4}}+\dfrac{\ep}{(n\ep/x_{\ep,n})^{1/2}}+\dfrac{\ep}{n^{1/2}}+\dfrac{\ep/x_{\ep,n}}{n^{1/2}}+\dfrac{\ep/x_{\ep,n}}{n\ep/x_{\ep,n}}+\dfrac{\ep/x_{\ep,n}}{n}+\dfrac{x_{\ep,n}}{(n\ep/x_{\ep,n})^{3/2}}+\\
+\dfrac{x_{\ep,n}}{n^{3/4}}+\dfrac{x_{\ep,n}}{n^{3/2}}+\dfrac{1}{n}+\dfrac{1}{(n\ep/x_{\ep,n})^2}+\dfrac{1}{n^2}\Bigr)
\times\mathcal{P}(\wt{u}_1,\wt{u}_2,\wt{t}_{\pm},\wt{s})\le 
\dfrac{\ep^{7/6}}{n^{1/4}x_{\ep,n}^{3/2}}\min\{x_{\ep,n}^{-1};n^{1/4}\}\cdot\mathcal{P}(\wt{u}_1,\wt{u}_2,\wt{t}_{\pm},\wt{s})
\end{split}
\end{equation*}
since $(1+c)\ep\le x_{\ep,n}\le C\ep^{1/3}$ and $\dfrac{1}{n}<\ep<\dfrac{1}{n^{1/3}}$.
Also observe that 
$|u_1+x_{\ep,n}|=\dfrac{|\wt{u}_1|}{(n\ep/x_{\ep,n})^{1/2}}$, which gives us
$$
\widehat{\Phi}_n(\wt{\mathbf{x}})\le \dfrac{\ep^{2/3}}{n^{3/4}x_{\ep,n}}\min\{x_{\ep,n}^{-1};n^{1/4}\}\cdot\mathcal{P}(\wt{u}_1,\wt{u}_2,\wt{t}_{\pm},\wt{s}).
$$
We also estimate 
$$
\Big|\int\limits_{\mathcal{R}_{1+}(t)} e^{-nr^2}\,dr\Big|\le |t|+Cn^{-1/2}\le C\max\{x_{\ep,n};n^{-1/4}\}(1+|\wt{t}_{\pm}|)
$$ 
for $t$ in the neighbourhood.
We are left to check that
$$
\dfrac{n^{5/2}\ep^{1/3}}{(n\ep/x_{\ep,n})^{2\cdot 1/2}n^{1/4+1/2}}\cdot \dfrac{\ep^{2/3}}{n^{3/4}x_{\ep,n}}\min\{x_{\ep,n}^{-1};n^{1/4}\}\cdot \max\{x_{\ep,n};n^{-1/4}\}
$$
is bounded, which is true.
\bigskip

\noindent
\textbf{\textit{Subcase 2b.}} Suppose $\ep> n^{-1/3}$. 
We have following expansions:
\begin{equation*}
\begin{split}
&F_{n,1}(u_1,u_2)=\L_n(x_{\ep,n}^2)-(x_{\ep,n}-\ep)^2-\kappa_1(u+x_{\ep,n})^2-\kappa_6u_2^2+O(|u_1+x_{\ep,n}|^3+|u_2|^3);\\
&F_{n,2}(t,s,r)=-\L_n(x_{\ep,n}^2)+(x_{\ep,n}-\ep)^2-(\kappa_3 \tau^2+\kappa_4\tau^3+\kappa_5 \tau^4+\kappa_2 s^2+r^2) +O(|s|^3+|\tau|^5),
\end{split}
\end{equation*}
where $t=ix_{\ep,n}+(\pm1+i(1+\ep-x_{\ep,n}))\tau$, $\kappa_1,\kappa_2,\kappa_3,\kappa_6$ are defined in \eqref{def_k_1}, \eqref{def_k_2}, \eqref{def_k_3}, \eqref{def_k_6} and $\kappa_4,\kappa_5$ satisfy
\eqref{def_k_4_k_5}.
It is easy to see that $\kappa_{1,6}\ge c\ep/x_{\ep,n}$, $\kappa_2\ge c$, $\Re\kappa_{3}\ge c\ep$, 
thus
\begin{equation*}
\begin{split}
&nF_{n,1}(u_1,u_2)\le nF_{n,1}(x_{\ep,n},0)-c\log^2 \dfrac{n\ep}{x_{\ep,n}}, \quad\text{when } \max\{|u_1+x_{\ep,n}|,|u_2|\}>(n\ep/x_{\ep,n})^{-1/2}\log (n\ep/x_{\ep,n}),\\
&n\Re F_{n,2}(t,s,r)\le nF_{n,2}(ix_{\ep,n},0,0)-c\log^2 n, \quad\text{when }|s|>n^{-1/2}\log n \text{ or } r\in\mathcal{R}_2(t), |r|>n^{-1/2}\log n,\\
&n\Re F_{n,2}(t,s,r)\le nF_{n,2}(ix_{\ep,n},0,0)-c\log^2 (n\ep), \quad\text{when } |t-ix_{\ep,n}|>(n\ep)^{-1/2}\log (n\ep).
\end{split}
\end{equation*}
Since $n\ep/x_{\ep,n}\ge cn^{7/9}$ and $n\ep\ge n^{2/3}$, 
the above bounds allow us to restrict the integration to the neighbourhood of the saddle point $u_1=-x_{\ep,n}$, $u_2=0$, $t=ix_{\ep,n}$, $s=0$, $r=0$. Make a change $u_1=-x_{\ep,n}+(n\ep/x_{\ep,n})^{-1/2}\wt{u}_1$, $u_2=(n\ep/x_{\ep,n})^{-1/2}\wt{u}_2$, $t=ix_{\ep,n}+(n\ep)^{-1/2}(\pm1+i(1+\ep-x_{\ep,n}))\wt{t}_{\pm}$, $s=n^{-1/2}\wt{s}$. Expanding $\wt{\varphi}$ up to the first order and using Remark~\ref{rem:phi_bounds}  one can obtain
\begin{equation*}
\begin{split}
|\wt\varphi(u_1,u_2,t,s)|\le C\Big(\dfrac{\ep^2}{x_{\ep,n}^2}+\dfrac{|\wt{u}_1|+|\wt{u}_2|}{(n\ep/x_{\ep,n})^{1/2}}+\dfrac{|\wt{t}_{\pm}|}{(n\ep)^{1/2}}+\dfrac{|\wt{s}|}{n^{1/2}}\Big)
\le \ep^{5/3}x_{\ep,n}^{-5/2}\cdot\mathcal{P}(\wt{u}_1,\wt{u}_2,\wt{t}_{\pm},\wt{s})
\end{split}
\end{equation*}
since $(1+c)\ep\le x_{\ep,n}\le C\ep^{1/3}$ and $\ep\ge \dfrac{1}{n^{1/3}}$. Also observe that 
$|u_1+x_{\ep,n}|=\dfrac{|\wt{u}_1|}{(n\ep/x_{\ep,n})^{1/2}}$, which gives us
$
\widehat{\Phi}_n(\wt{\mathbf{x}})\le n^{-1/2}x_{\ep,n}^{-2}\ep^{7/6}\,\mathcal{P}(\wt{u}_{1},\wt{u}_2,\wt{t}_{\pm},\wt{s}).
$
Next we estimate the integral with respect to $r$: 
$$
\Big|\int\limits_{\mathcal{R}_{1+}(t)} e^{-nr^2}\,dr\Big|\le |t|+Cn^{-1/2}\le C(x_{\ep,n}+(n\ep)^{-1/2}|\wt{t}_{\pm}|+n^{-1/2})\le Cx_{\ep,n}(1+|\wt{t}_{\pm}|).
$$
We are left to prove that 
$
\dfrac{n^{5/2}\ep^{1/3}}{(n\ep/x_{\ep,n})^{2\cdot 1/2}(n\ep)^{1/2}n^{1/2}}\cdot Cn^{-1/2}x_{\ep,n}^{-2}\ep^{7/6}\cdot Cx_{\ep,n}
$
is bounded, which is true.

\end{proof}

Finally, we consider the case when $z$ lies outside of $D$ far enough from $\partial D$.

\begin{pro}\label{pro:uni_bound_far_out_lar}
Fix an arbitrary $d>0$, and set $E_{d,out}:=\{z\in E\mid z\notin D,\ \dist(z,\partial D)\ge d\}$. Then there exist $n_0\in\mathbb{N}$ and $C>0$ such that 
$$
\ep^{1/2}\,|T(\ep,n,z)|\le C
$$ 
for $n\ge n_0$, $n^{-1}<\ep<\ep_0$ and $z\in E_{d,out}$.
\end{pro}

\begin{proof}
Let $x_{\ep,n}$ be the positive root of~\eqref{eq_x_en}. According to Proposition~\ref{pro:sols_bound_ext}, we have $(1+c)\ep\le x_{\ep,n}\le C\ep$ for some $c,C>0$. Next one can simply repeat the argument from Case 2 of the proof of Proposition~\ref{pro:uni_bound_close_lar} considering Subcase 2a and Subcase 2b and improving the bounds using $(1+c)\ep\le x_{\ep,n}\le C\ep$.

\end{proof}

It is easy to see that Theorem~\ref{thm:uni_bound_far} and Theorem~\ref{thm:uni_bound_near} now follow from Propositions~\ref{pro:uni_bound_far_in_sm}--\ref{pro:uni_bound_far_out_lar}.

\section{Rate of convergence}\label{sc:rate}

In this section we present the proof of Theorem~\ref{thm:rate_of_conv}. In order to estimate the difference
$$
\Bigl|\E\Bigl\{\dfrac{1}{n}\sum_{j=1}^n h(z_j)\Bigr\}-\int h(z)\rho(z)\,\dz\Bigr|,
$$
we fix $\epsilon=n^{-1/2}$ and approximate $\rho(z)$ by $\bar{\rho}_{\ep,n}(z)$. It is sufficient to obtain the following bounds:
\begin{align}
\label{inq1}
&\Bigl|\E_{H_n}\Bigl\{\dfrac{1}{n}\sum_{j=1}^n h(z_j)\Bigr\}-\int h(z)\bar{\rho}_{\ep,n}(z)\,\dz\Bigr|\le Cn^{-1/2}\qquad \text{unifromly in $n$};\\
\label{inq2}
&|\bar{\rho}_{\ep,n}(z)-\wh\rho_{\ep,n}(z)|\le Cn^{-1/2}\qquad \text{unifromly in $n$, $z\in \supp h$};\\
\label{inq3}
&|\wh{\rho}_{\ep,n}(z)-\rho_{\ep}(z)|\le Cn^{-1/2}\qquad \text{unifromly in $n$, $z\in \supp h$};\\
\label{inq4}
&|\rho_{\ep}(z)-\rho(z)|\le Cn^{-1/2}\qquad \text{unifromly in $n$, $z\in \supp h$},
\end{align}
where $\bar{\rho}_{\ep,n}(z)$, $\rho_{\ep}(z)$ and $\rho(z)$ are given by \eqref{asympt_rho_ep,n}, \eqref{rho_ep_form}, \eqref{rho_main_form} respectively and $\wh{\rho}_{\ep,n}(z)$ is the main asymptotic term of $\bar{\rho}_{\ep,n}(z)$:
$$
\wh{\rho}_{\ep,n}(z)=\dfrac{1}{\pi}\Bigl(\dfrac{|\tr(A_n-z)G^2(x_{\ep,n}^2)|^2}{\tr G^2(x_{\ep,n}^2)+\ep/(2x_{\ep,n}^3)}+x^2_{\ep,n}\cdot\tr G(x_{\ep,n}^2)\wt{G}(x_{\ep,n}^2)\Bigr).
$$
One can easily show that \eqref{inq3} follows from condition (C5). Furthermore, Proposition~\ref{pro:sols_bound_int} yields that $|x_{\ep}-x_0|\le C\ep$, hence $|\rho_{\ep}(z)-\rho(z)|\le C\ep=Cn^{-1/2}$, which gives us~\eqref{inq4}. We are left to check~\eqref{inq1} and \eqref{inq2}.

In order to get~\eqref{inq1}, observe that $\lim\limits_{\ep\to 0}\int h(z)\bar{\rho}_{\ep,n}(z)\,\dz=\E_{H_n}\Bigl\{\dfrac{1}{n}\sum_{j=1}^n h(z_j)\Bigr\}$ uniformly in $n$ according to Proposition~\ref{l:uni_conv}, and Theorem~\ref{thm:uni_bound_far} shows that 
\begin{equation*}
\Bigl|\E_{H_n}\Bigl\{\dfrac{1}{n}\sum_{j=1}^n h(z_j)\Bigr\}-\int h(z)\bar{\rho}_{\ep,n}(z)\,\dz\Bigr|\le C\ep^{1/2}=Cn^{-1/4}.
\end{equation*}
However, the bound can be improved in the following way. We are now in the case when $z\in D$, $\dist(z,\partial D)\ge d$ and $\ep=n^{-1/2}>n^{-1}$. This case is covered by Proposition~\ref{pro:uni_bound_far_in_lar}, and the proof of this proposition actually gives an inequality
$$
|T(\ep,n,z)|\le C\left(1+\dfrac{1}{n^{1/2}\ep}\right),
$$
which means $|T(n^{-1/2},n,z)|\le C$. Hence, for $\ep=n^{-1/2}$ we obtain a uniform bound $|\mathrm{\Phi}(\ep,n,z)|\le C\ep=Cn^{-1/2}$, which implies~\eqref{inq1}.
 
Inequality~\eqref{inq2} is just a bound on the error term of $\rho_{\ep,n}(z)$ when $\ep=n^{-1/2}$. Now we adjust the argument in Subsection~\ref{sbsc:<ge^nF>} and Proposition~\ref{pro:for_I_1}. We need to perform a bit more accurate asymptotic analysis of $\bar{\rho}_{\ep,n}(z)$ since now $\ep$ is not fixed but depends on $n$.

As before, Proposition~\ref{pro:int_rep_dens} yields that $
\bar{\rho}_{\ep,n}(z)=\dfrac{1}{\pi}(\I_1+\I_2+\I_3+\I_4)
$
with $\I_k$ defined in \eqref{def_I_k}. It is easy to show that $\I_4=O(n^{-1})$, so we are interested in bounds on the error terms of $\I_1,\I_2,\I_3$.

We start with the integrals $\I_{g_n}$ which are defined in~\eqref{def_I_g_n} with $g_n(x,y)$ satisfying the conditions in Subsection~\ref{sbsc:<ge^nF>} together with the additional condition:
\begin{equation}\label{cond_at_saddle}
g_n(x_{\ep,n}^2,x_{\ep,n}^2)=O(\ep).
\end{equation} 
According to the proof of Proposition~\ref{pro:for_I_k} 
we have $\I_2=\I_{g_{n,2}}$ and $\I_3=\I_{g_{n,3}}$ for $g_{n,2}(x,y)=
x\cdot\tr G(x)\wt{G}(x)\cdot\varphi(x,y,z,z)$ and $g_{n,3}(x,y)=\Bigl(\tr\,(z-A_n)G(x)-\tr\,(z-A_n)G(y)\Bigr)\partial_{z_1}\varphi(x,y,z,z)$. Condition~\eqref{cond_at_saddle} holds for $g_{n,2}$ since $\varphi(x_{\ep,n}^2, x_{\ep,n}^2)=O(\ep)$ which can be obtained from straightforward computation. This condition also holds for $g_{n,3}$ simply because $g_{n,3}(x,x)=0$.

Recall from Subsection~\ref{sbsc:<ge^nF>} that
\begin{equation*}
\I_{g_n}=\frac{8n^3}{\pi^3}\int_{\mathcal{V}} \Phi_n(\mathbf{x})\,e^{nF_n(\mathbf{x})}\,d\mathbf{x},
\end{equation*}
where
\begin{equation*}
\begin{split}
&F_{n,1}(u)=\L_n(u^2)-(u-\ep)^2,\quad F_{n,2}(t,s,r)=-\Bigl(\L_n(s^2-t^2)+(t-i\ep)^2+r^2\Bigr),\\
&F_n(\mathbf{x})=F_{n,1}(u)+F_{n,2}(t,s,r)-2\ep w^2-\ep v^2,\quad 
\Phi_n(\mathbf{x})=\dfrac{r-it-\ep}{\sqrt{v^2+4r-4it-4\ep}}\cdot\dfrac{u}{\sqrt{2u-w^2}}\cdot g_n(u^2,s^2-t^2).
\end{split}
\end{equation*}
We can change contours as in Subsection~\ref{sbsc:<ge^nF>}, since $x_{\ep,n}\ge \kappa_0>0$. It is easy to check that
\begin{equation*}
\begin{split}
&nF_{n,1}(u)\le nF_{n,1}(x_{\ep,n})-c\log^2 n, \quad\text{when } |u-x_{\ep,n}|>n^{-1/2}\log n,\\
&n\Re F_{n,2}(t,s,r)\le nF_{n,2}(ix_{\ep,n},0,0)-c\log^2 n, \quad\text{when }\max\{|t-ix_{\ep,n}|,|s|,|r|\}>n^{-1/2}\log n,\\
&-n\ep v^2-2n\ep w^2\le -c\log^2(n\ep), \quad\text{when } \max\{|v|,|w|\}>(n\ep)^{-1/2}\log (n\ep).
\end{split}
\end{equation*}
Since $n\ep=n^{1/2}$, we can restrict the integration to a neighbourhood 
$$\{\mathbf{x}\in \wt{\mathcal{V}}\colon |u-x_{\ep,n}|,\, |s|,\, |t-ix_{\ep,n}|,\, |r| <n^{-1/2}\log n, |v|<(n\ep)^{-1/2}\log (n\ep), 0\le w<(n\ep)^{-1/2}\log (n\ep)\}.
$$
Make a change $u=x_{\ep,n}+n^{-1/2}\wt u$, $t=ix_{\ep,n}+n^{-1/2}\wt t$, $s=n^{-1/2}\wt s$, $r=n^{-1/2}\wt r$, $v=(n\ep)^{-1/2}\wt{v}$, $w=(n\ep)^{-1/2}\wt{w}$, then the coefficient before the integral becomes $\dfrac{C}{\ep}$. Expand the multipliers of the integrand:
\begin{equation}\label{r o c expan}
\begin{split}
&\dfrac{r-it-\ep}{\sqrt{v^2+4r-4it-4\ep}}=
\dfrac{x_{\ep,n}^{1/2}}{2}+n^{-1/2}\mathcal{P}_1(\wt r, \wt t)+n^{-1}\mathcal{P}_2(\wt r, \wt t)+(n\ep)^{-1}a_1\wt v^2+O(n^{-1}\log^k n);\\
&\dfrac{u}{\sqrt{2u-w^2}}=\Bigl(\dfrac{x_{\ep,n}}{2}\Bigr)^{1/2}+n^{-1/2}b_1\wt u+n^{-1} b_2\wt u^2+(n\ep)^{-1}b_3\wt w^2+O(n^{-1}\log^k n);\\
&g_n(u^2,s^2-t^2)=g_n(x_{\ep,n}^2,x_{\ep,n}^2)+n^{-1/2}\mathcal{P}_1(\wt u, \wt t)+n^{-1}\mathcal{P}_2(\wt u, \wt t,\wt s)+O(n^{-3/2}\log^k n);\\
&e^{nF_n(\mathbf{x})}=\Bigl(1+n^{-1/2}\mathcal{P}_3(\wt u, \wt t, \wt s)+n^{-1}\mathcal{P}_4(\wt u, \wt t, \wt s)+O(n^{-3/2}\log^k n)\Bigr)e^{-\kappa_1 \wt{u}^2-\kappa_1\wt{t}^2-\kappa_2\wt{s}^2-\wt{r}^2-\ep \wt{v}^2-2\ep\wt{w}^2},
\end{split}
\end{equation}
where $a_j,b_j$ are some bounded coefficients and $\mathcal{P}_j$ are some homogeneous polynomials of degree $j$ with bounded coefficients (we are not interested in the exact form of these polynomials, so we may use the same notation for different polynomials). 

After multiplying such expansions, we need to track only the monomials of even total degree, since the monomials of odd degree vanish after the integration. Using $g_n(x_{\ep,n}^2,x_{\ep,n}^2)=O(\ep)$ one can check that the monomials of even degree have coefficients of order $O(n^{-1})$, while the terms coming from the error terms of the expansions are of order $O(n^{-3/2}\log^k n)$. Taking into account the multiplier $\dfrac{C}{\ep}$ before the integral, we deduce that the error term of $\I_{g_n}$ has order $O((n\ep)^{-1})=O(n^{-1/2})$.

We obtained that the error terms of $\I_2$, $\I_3$ are $O(n^{-1/2})$ when $\ep=n^{-1/2}$. Next we study $\I_1$.
According to \eqref{I_1_int_form}, we have
\begin{equation*}
\I_{1}=Cn^4\int_{\mathcal{V}} \Phi_n(\mathbf{x})\,e^{nF_n(\mathbf{x})}\,d\mathbf{x},
\end{equation*}
where
\begin{equation*}
\begin{split}
&F_{n,1}(u)=\L_n(u^2)-(u-\ep)^2,\quad F_{n,2}(t,s,r)=-\Bigl(\L_n(s^2-t^2)+(t-i\ep)^2+r^2\Bigr),\\
&F_n(\mathbf{x})=F_{n,1}(u)+F_{n,2}(t,s,r)-2\ep w^2-\ep v^2,\\
&\Phi_n(\mathbf{x})=
\bigl(\,\ol{p(u_1^2+u_2^2)}-\ol{p(x_{\ep,n}^2)}\,\bigr)\,\Bigl( \bigl(p(u^2)-p(x_{\ep,n}^2)\bigr)-\bigl(p(s^2-t^2)-p(x_{\ep,n}^2)\bigr)\Bigr)\cdot\\
&\cdot\dfrac{r-it-\ep}{\sqrt{v^2+4r-4it-4\ep}}\cdot\dfrac{u}{\sqrt{2u-w^2}}\cdot \varphi(u^2,s^2-t^2).
\end{split}
\end{equation*}
We can make the same change of contours and variables as for $\I_{g_n}$ above. The coefficient before the integral becomes $C\dfrac{n}{\ep}$. The multipliers $\dfrac{r-it-\ep}{\sqrt{v^2+4r-4it-4\ep}}$, $\dfrac{u}{\sqrt{2u-w^2}}$ and $e^{nF_n(\mathbf{x})}$ have the same expansions as in~\eqref{r o c expan}, while
\begin{equation*}
\begin{split}
&\ol{p(u_1^2+u_2^2)}-\ol{p(x_{\ep,n}^2)}=n^{-1/2}a_1\wt u+n^{-1}a_2\wt u^2+O(n^{-3/2}\log^k n);\\
&\bigl(p(u^2)-p(x_{\ep,n}^2)\bigr)-\bigl(p(s^2-t^2)-p(x_{\ep,n}^2)\bigr)=n^{-1/2}\mathcal{P}_1(\wt u, \wt t, \wt s)+n^{-1}\mathcal{P}_2(\wt u, \wt t, \wt s)+O(n^{-3/2}\log^k n);\\
&\varphi(u^2,s^2-t^2)=\varphi(x_{\ep,n}^2,x_{\ep,n}^2)+n^{-1/2}\mathcal{P}_1(\wt u, \wt t, \wt s)+n^{-1}\mathcal{P}_2(\wt u, \wt t,\wt s)+O(n^{-3/2}\log^k n),
\end{split}
\end{equation*}
with $\varphi(x_{\ep,n}^2,x_{\ep,n}^2)=O(\ep)$. Using similar argument to the one for $\I_{g_n}$ above, we obtain that the error term of $\I_1$ also has order $O((n\ep)^{-1})=O(n^{-1/2})$.

\begin{rema}
The same method works for an arbitrary $h\in C_c^2(\mathbb{C})$, not only supported inside the bulk $D$. However, for $z$ close to the edge $\partial D$ the constant $\alpha$ in the rate $n^{-\alpha}$ obtained by such method becomes much worse.

\end{rema}

\end{document}